\documentclass{sig-alternate-05-2015}

\usepackage{graphicx}
\usepackage{amsmath}
\usepackage{amssymb}
\usepackage{url}
\usepackage{xspace}
\usepackage{balance}
\usepackage{algorithm}	
\usepackage{algpseudocode}
\usepackage{multirow}
\usepackage{subfig}
\usepackage{fixltx2e}
\usepackage{times}
\usepackage{enumitem}

\graphicspath{{Results_Project2/Results/Plots/}}

\setlist{nolistsep}

\title{K-Dominant Skyline Join Queries: Extending the Join Paradigm to
K-Dominant Skylines}

\numberofauthors{1}

\author{
\alignauthor
Anuradha Awasthi \quad
Arnab Bhattacharya \quad
Sanchit Gupta \quad
Ujjwal Kumar Singh\\
\affaddr{Dept. of Computer Science and Engineering, Indian Institute of Technology, Kanpur, India}
\email{\{aawasthi,arnabb,sanchitg,ujjkumsi\}@cse.iitk.ac.in}
}

\newtheorem{thm}{Theorem}
\newtheorem{defn}{Definition}
\newtheorem{lem}{Lemma}

\newtheorem{prob}{Problem}
\newtheorem{assmp}{Assumption}
\newtheorem{obs}{Observation}

\renewcommand{\comment}[1]{}

\algrenewcommand{\algorithmicrequire}{\textbf{Input:}}
\algrenewcommand{\algorithmicensure}{\textbf{Output:}}
\algrenewcommand{\algorithmicforall}{\textbf{for each}}

\newcommand{\tabcaption}[1]{\vspace*{-2mm}\caption{#1}\vspace*{-2mm}}
\newcommand{\figcaption}[1]{\vspace*{-2mm}\caption{#1}\vspace*{-5mm}}

\newcommand{\kdom}{\ensuremath{k}-dominant\xspace}

\renewcommand{\ss}{\ensuremath{SS}\xspace}
\newcommand{\sn}{\ensuremath{SN}\xspace}
\newcommand{\nn}{\ensuremath{NN}\xspace}

\newcommand{\succk}{\ensuremath{\succ_k}\xspace}

\newcommand{\koned}{\ensuremath{k'_1}\xspace}
\newcommand{\ktwod}{\ensuremath{k'_2}\xspace}
\newcommand{\konedd}{\ensuremath{k''_1}\xspace}
\newcommand{\ktwodd}{\ensuremath{k''_2}\xspace}

\newcommand{\ts}{\ensuremath{\tau}\xspace}

\newcommand{\subfigwidth}{0.52\columnwidth}

\isbn{XXX}
\doi{XXX}

\begin{document}

\maketitle

\begin{abstract}
	Skyline queries enable multi-criteria optimization by filtering objects
	that are worse in all the attributes of interest than another object.  To
	handle the large answer set of skyline queries in high-dimensional
	datasets, the concept of $k$-dominance was proposed where an object is said
	to dominate another object if it is better (or equal) in at least $k$
	attributes.  This relaxes the full domination criterion of normal skyline
	queries and, therefore, produces lesser number of skyline objects.  This is
	called the $k$-dominant skyline set.  Many practical applications, however,
	require that the preferences are applied on a joined relation.  Common
	examples include flights having one or multiple stops, a combination of
	product price and shipping costs, etc.  In this paper, we extend the
	$k$-dominant skyline queries to the join paradigm by enabling such queries
	to be asked on joined relations.  We call such queries \emph{KSJQ
	($k$-dominant skyline join queries)}.  The number of skyline attributes,
	$k$, that an object must dominate is from the combined set of skyline
	attributes of the joined relation.  We show how pre-processing the base
	relations helps in reducing the time of answering such queries over the
	na\"ive method of joining the relations first and then running the
	$k$-dominant skyline computation.  We also extend the query to handle cases
	where the skyline preference is on aggregated values in the joined relation
	(such as total cost of the multiple legs of the flight) which are available
	only after the join is performed.  In addition to these problems, we devise
	efficient algorithms to choose the value of $k$ based on the desired
	cardinality of the final skyline set.  Experiments on both real and
	synthetic datasets demonstrate the efficiency, scalability and practicality
	of our algorithms.
\end{abstract}

\keywords{Skyline query; K-dominant skyline query; Join; Aggregation;
K-dominant skyline join query; KSJQ}

\section{Introduction and Motivation}

Skyline queries are widely used to enable multi-criteria decision making in
databases \cite{skylinebnl}.  Consider a scenario where a person wants to buy a
good house.  Her preferences are low cost, proximity to market, quiet
neighborhood, etc.  In a real scenario, it is almost impossible to find a
single house that is best in all her preferences.  The skyline query helps her
narrow down the choices by filtering out houses that are \emph{worse} (or
equal) than some other house in \emph{all} the preferences.  Assuming
rationality, her choices cannot lie outside the \emph{skyline set}.

In high-dimensional spaces, however, the skyline set becomes less useful due to
its impractically large size.  The size tends to increase exponentially with
dimensionality \cite{skylinekernel}.  This happens as it becomes harder for any
object to dominate another object in \emph{all} the attributes.  The problem is
especially severe in real datasets where the data is generally anti-correlated
in nature.  For example, a quiet neighborhood closer to a market is likely to
be more costly.

There have been various works on handling the large cardinality of skyline sets
in high dimensions.  The most prominent is that of \kdom skylines where,
instead of being better in all the $d$ dimensions, an object need only be
better in some $k < d$ dimensions to dominate another object
\cite{skylinekdominant}.  As a result, it becomes easier for an object to be
dominated which leads to lesser number of skylines.

The \kdom skylines are quite useful in real scenarios.  In the example of
housing discussed above, if there are many attributes, it may be rare
that one house is better than another on all the counts.  Instead, if a smaller
number of attributes, say $k = 2$, is specified, there are more chances of
finding a house that has a lower cost and a quieter neighborhood (but may not
be closer to a market) than another house.  As a result, more houses can be
filtered, and the retrieved skyline set becomes more manageable and useful.

To the best of our knowledge, however, the \kdom skyline queries have not been
explored for \emph{multiple} relations.\footnote{A preliminary version of this
paper will appear as a poster \cite{poster}.}  Suppose there are two relations
having $d_1$ and $d_2$ skyline attributes.  After joining, a bigger relation
with $d_1 + d_2$ skyline attributes is formed.  (We discuss the different
variants and restrictions later.)  A $k$-dominant skyline, where $k < d_1 +
d_2$, is then sought on this joined relation.

A real-life example of this situation happens often in flight bookings.
Suppose a person wants to fly from city A to city B.  Her preferences are lower
cost, lower duration, higher ratings and higher amenities.  While the basic
skyline works for direct flights, in many cases, a flight route from A to B
includes one (or more) stopovers.  Thus, a valid flight path contains the join
of all flights from city A to other cities and from those cities to city B
where the intermediate city is the same.  The preferences a user would want now
applies to the \emph{entire} flight path and not a single leg of the journey.
The skyline is, therefore, needed on the joined relation
\cite{asjq,skylinejoin}.

Once more, it is harder for a flight combination to dominate another flight
combination in all the skyline attributes over the two relations.  Here, the
\kdom skyline query is a natural choice, where $k$ is less than the total
number of the skyline attributes in the joined relation.

With the increase in the number of attributes, the size of the skyline set
increases further and, hence, computing the \kdom skylines becomes even more
relevant and useful.

The na\"ive approach first creates the joined relation and then subsequently
computes the $k$-dominance.  This strategy, while straightforward, is
inefficient and impractical for large datasets.

A further practical consideration in the flight example is that a user is not
really bothered about the individual legs, but rather the \emph{total} cost and
\emph{total} duration of the journey.  Thus, the skyline preferences should be
applied on the \emph{aggregated} values of attributes from the base relations.
Note that the aggregated values are available only after the join and are,
therefore, harder to process efficiently.

In this paper, we explore the question of finding \kdom skylines on joined
relations, where the skyline preferences can be on both aggregated and
individual values.  Apart from posing the problem, our main contribution is to
push the skyline operator before the join as much as possible, thereby making
the whole algorithm efficient and practical.

In addition to finding \kdom skylines, an important question that often arises
in practical applications is how to choose a ``good'' value of $k$?  While
there is no universal answer, one of the guiding principles is the number of
skyline objects finally returned \cite{skylinekdominant}.  The basic idea of
skyline queries is to serve as a filter for poor objects and, thus, a user may
find it easier to specify a value of $\delta$ objects that she is interested in
examining more thoroughly rather than a value of $k$.  Since the size of the
skyline set increases with $k$, the ``optimal'' value of $k$ may be then taken
as the smallest one that returns at least $\delta$ skyline objects, or the
largest one that returns at most $\delta$ skyline objects.

We address the above question in the context of \kdom skylines over joined
relations.

In sum, our contributions are:

\begin{enumerate}

	\item We propose the problem of finding \kdom skyline queries over joined
		relations.  We term such queries \emph{KSJQ}.

	\item We design efficient algorithms to solve the KSJQ problem.

	\item We devise ways to arrive at a good value of ``$k$'' by specifying a
		threshold size of the \kdom skyline set.

\end{enumerate}

The rest of the paper is organized as follows.  Sec.~\ref{sec:background} sets
the background on joins of \kdom skylines.  Using this, different problem
statements are defined in Sec.~\ref{sec:problem}.  Sec.~\ref{sec:work} outlines
the related work.  Various optimizations that can improve the efficiency of the
problem are explained in Sec.~\ref{sec:opt}.  Algorithms that use these
optimizations are described in Sec.~\ref{sec:algo}.  Sec.~\ref{sec:exp}
analyzes the experimental results before Sec.~\ref{sec:conc} concludes.

\section{Background}
\label{sec:background}

\subsection{Skylines}

Consider a dataset $R$ of objects.  For each object $u$, a set of $d$
attributes $\{u_1, \dots, u_d\}$ are specified, which form the \emph{skyline
attributes}.  For each of the skyline attributes, without loss of generality,
the preference is assumed to be less than ($<$), i.e., a lower value is
preferred over a higher one.  An object $u$ \emph{dominates} another object
$v$, denoted by $u \succ v$, if and only if, for all the $d$ skyline
attributes, $u_i$ is preferred over or equal to $v_i$, and there exists at
least one skyline attribute where $u_j$ is strictly preferred over $v_j$.  The
\emph{skyline} set $S \subseteq R$ contains objects that are \emph{not}
dominated by any other object \cite{skylinebnl}.  In other words, every object
in the non-skyline set is dominated by at least one object in the dataset.

\subsection{K-Dominant Skylines}

The \emph{$k$-dominant} skyline query \cite{skylinekdominant} relaxes the
definition of domination between objects.  An object $u$ \emph{$k$-dominates}
another object $v$, denoted by $u \succ_k v$, if and only if, for at least $k$
of the $d$ skyline attributes, $u_i$ is preferred over or equal to $v_i$, and
there exists at least one skyline attribute where $u_j$ is strictly preferred
over $v_j$.  The \emph{\kdom skyline} set contains objects that are not
$k$-dominated by any other object.  The $k$ attributes are \emph{not} fixed and
can be \emph{any} subset of the $d$ skyline attributes.

The \kdom query is particularly problematic when $k \leq \frac{d}{2}$ since
then two objects can dominate each other.  Even when $k > \frac{d}{2}$, the
$k$-dominance relationship is \emph{not transitive} and may be even
\emph{cyclic}: $u \succ_k v \succ_k w \succ_k u$.

\subsection{Multi-Relational Skylines}

Consider two datasets $R_1$ and $R_2$ with $d_1$ and $d_2$ skyline attributes
respectively.  The \emph{join} of the two relations (using appropriate join
conditions) forms the dataset $R = R_1 \Join R_2$ with $d = d_1 + d_2$ skyline
attributes.  The \emph{multi-relational skyline} is extracted from $R$
\cite{skylinejoin}.

In a significant variant of the problem, a number of skyline attributes in each
relation are marked for \emph{aggregation} with corresponding attributes from
the other relation \cite{asjq}.  Thus, $a$ attributes out of $d_1$ in $R_1$ and
$d_2$ in $R_2$ are aggregated in the joined relation $R$.  As a result, $R$
contains $(d_1 - a) + (d_2 - a) + a = d_1 + d_2 - a$ skyline attributes.  The
skyline query is then asked over these attributes.  The aggregation function is
assumed to be \emph{monotonic}; this ensures that if the base values of two
tuples $u \in R_1$ and $t \in R_2$ are preferred over the base tuples of two
other tuples $v \in R_1$ and $s \in R_2$ respectively, the aggregated value of
$u \Join t \in R$ will be necessarily preferred over the aggregated value of $v
\Join s \in R$.

The case for more than two base relations can be handled by cascading the
joins.

\section{Problem Statements}
\label{sec:problem}

We propose two main variants of what we call the \textsc{k-dominant skyline
join query (KSJQ)} problem.

The first variant works on relations for which the skyline attributes are
strictly \emph{local}.


The number of skyline attributes in the first and second relations, $R_1$ and
$R_2$, are $d_1$ and $d_2$ respectively.
\begin{align}
	R_1 &= \{h_{1_1}, \dots, h_{1_m}, s_{1_1}, \dots, s_{1_{d_1}}\} \\
	R_2 &= \{h_{2_1}, \dots, h_{2_m}, s_{2_1}, \dots, s_{2_{d_2}}\} \\
	\label{eq:basic}
	R &= \{h_{1}, \dots, h_{m}, s_{1_1}, \dots, s_{1_{d_1}}, s_{2_1}, \dots, s_{2_{d_2}}\}
\end{align}
where $h_{i}$ captures the join of $h_{1_i}$ with $h_{2_i}$ in the joined
relation $R = R_1 \Join R_2$, and $s_{1_i}, s_{2_i}$ are the skyline attributes.


\begin{prob}[KSJQ]
	\label{prob:plain}
	Given two datasets $R_1$ and $R_2$ having $d_1$ and $d_2$ skyline attributes
	respectively, find \emph{\kdom skylines} from the joined relation $R = R_1
	\Join R_2$ having $d = d_1 + d_2$ skyline attributes.
\end{prob}

We restrict $k$ to be at least one more than the dimensionality in the base
relations, i.e., $\max\{d_1, d_2\} < k < d$.  This restriction constrains at
least some skyline attributes from each relation to satisfy the preferences.
In other words, the $k$ preferred attributes must span \emph{both} the
relations.  However, there is no restriction over how many skyline attributes
from each relation must be satisfied.  Denoting the number of skyline
attributes that are chosen from the two relations by $k_1$ and $k_2$ where $k_1
+ k_2 = k$, this implies that $1 \leq k_1 \leq d_1$ and $1 \leq k_2 \leq d_2$.

The second variant works on the \emph{aggregate} version of the problem.
Assuming a total of $a$ aggregate attributes, the number of attributes on which
only local skyline preferences are applied are $d_1 - a = l_1$ and $d_2 - a =
l_2$ respectively.  The joined relation $R$, therefore, contains $l_1 + l_2$
local attributes and $a$ aggregate attributes.

\comment{

Summarizing,
\begin{align}
	R_1 &= \{h_{1_1}, \dots, h_{1_m}, s_{1_1}, \dots, s_{1_{l_1}}, g_{1_1}, \dots, g_{1_a}\} \\
	R_2 &= \{h_{2_1}, \dots, h_{2_m}, s_{2_1}, \dots, s_{2_{l_2}}, g_{2_1}, \dots, g_{2_a}\} \\
	\label{eq:agg}
	R &= \{h_{1}, \dots, h_{m},
		s_{1_1}, \dots, s_{1_{l_1}}, s_{2_1}, \dots, s_{2_{l_2}},
		g_{1}, \dots, g_{a}\}
\end{align}
where, in addition to the joins $h_{i}$, the attributes $g_{j} = g_{1_j}
\oplus_j g_{2_j}$ denote the aggregated values.

}

\begin{prob}[Aggregate KSJQ]
	\label{prob:agg}
	Given two datasets $R_1$ and $R_2$ having $d_1$ and $d_2$ skyline attributes
	respectively, of which $a$ attributes are used for aggregation, find
	\emph{\kdom skylines} from the joined relation $R = R_1 \Join R_2$ having
	$d_1 + d_2 - a$ skyline attributes.
\end{prob}

Once more, we assume that $\max\{d_1, d_2\} < k \leq d$.
Expressing in terms of local and aggregate attributes, $\max\{l_1,
l_2\} + a < k \leq l_1 + l_2 + a$.

The third and fourth problems address the tuning of the value of $k$.
The value can be tuned in two ways.

\begin{prob}[At least $\delta$]
	\label{prob:k}
	Given a KSJQ query framework and a threshold number of skyline objects
	$\delta$, determine the \emph{smallest} value of $k$ that returns \emph{at
	least} $\delta$ skyline objects.
\end{prob}

\begin{prob}[At most $\delta$]
	\label{prob:kup}
	Given a KSJQ query framework and a threshold number of skyline objects
	$\delta$, determine the \emph{largest} value of $k$ that returns \emph{at
	most} $\delta$ skyline objects.
\end{prob}

Prob.~\ref{prob:kup} is directly linked with Prob.~\ref{prob:k}.  If $k^*$ is
the answer to Prob.~\ref{prob:k}, then the answer to Prob.~\ref{prob:kup} must
be $k^* - 1$.  There are two corner cases.  First, if $k^* = 1$, then the
answer to Prob.~\ref{prob:kup} should be trivially $1$ as well.  Second, if
$k^*$-dominant skyline returns \emph{exactly} $\delta$ skyline objects (or $k^*
= d$), then the answer to Prob.~\ref{prob:kup} is $k^*$ as well.  Thus,
henceforth, we focus only on Prob.~\ref{prob:k}.

\section{Related Work}
\label{sec:work}

The skyline operator was introduced in databases by \cite{skylinebnl} by
adopting the maximal vector or the Pareto optimal problem \cite{pareto}.
Several indexed \cite{skylinenn}, \cite{skylinebbs} and non-indexed algorithms
\cite{skylinebnl,skylinesfs,skylinebitmap} have been since proposed to retrieve
the skyline set.

Analyses of the cardinality of the skyline set \cite{skylinecardinality,
skylinekernel} have shown that the number of skyline objects can grow
exponentially with the increase in dimensionality.  Consequently, several
attempts have been made to restrict the size of the skyline set.  These mostly
include notions of approximate skylines or representative skylines
\cite{kmostrep2,kmostrep1,lskydiv,thicksky,adr,kmostrep,approxcontour,skydiver,epsdom,strongsky}.

A completely different approach---\emph{\kdom skylines}---was proposed by
\cite{skylinekdominant}.  It uses subsets of skyline attributes to control the
cardinality of the skyline set.  Although the parameter $k$ is an input to the
problem, the authors also proposed a way to derive the smallest $k$ that will
guarantee a $\delta$ number of skylines.

There have been many recent works on \kdom skylines including extensions to
high-dimensional spaces \cite{extendedkdom,kdompower}, using parallel
processing \cite{kdommapreduce}, cardinality estimation \cite{kdomcardinality}
as well as in update-heavy datasets \cite{kdomupdate} and combined datasets
\cite{kdomcombined}.

The skyline join problem where skylines are retrieved from joined relations,
both components of which contain skyline attributes, was introduced in
\cite{skylinejoin}.  This work, however, simply assumed that the skyline
attributes remain unchanged from the base relations to the joined one.  The
aggregate skyline join queries \cite{asjq} removed this restriction by enabling
skyline preferences to be posed on aggregate values of attributes formed by
combining attributes from the base relations.  The aggregation function was
assumed to be monotonic.

In this paper, we pose the \kdom skyline problem in the join paradigm
(including the aggregated version).  To the best of our knowledge, this is the
first work in this direction.

\section{Optimizations}
\label{sec:opt}

In this section, we describe the various optimization schemes that can be used
to speed-up the process of finding \kdom skylines in joined relations.
Sec.~\ref{sec:algo} uses these optimizations to design the algorithms.

\subsection{Join Attributes}

We assume an \emph{equality} join condition.  Two base tuples $u \in R_1$ and
$v \in R_2$ can be joined to form $t = u \Join v \in R$ if and only if all
their join attributes match.  Referring to Eq.~\eqref{eq:basic}, $\forall_{j =
1, \dots, m} \ h_{1_j} = h_{2_j}$.  (We discuss the relaxation of this
assumption in Sec.~\ref{sec:nonequality}.)

\begin{assmp}[Equality Join]
	The join conditions are all based on equality, i.e., the join is an equality
	join.
\end{assmp}

All the tuples in a base relation are, thus, divided into \emph{groups}
according to the value of the join attributes.  In every group, the values of
the join attributes, $h_{i_1}, \dots, h_{i_m} (i = 1,2)$, are the same.

In the flight example, since the joining criterion is destination of first
flight to be the source of the second flight, the first base relation is
divided on the basis of destinations while the second one is divided on
sources.

\subsection{Grouping}

Based on the $k$-dominance properties within each group, each base relation is
partitioned into different sets as follows.

A tuple may be a \kdom tuple for the entire base relation.  However, even when
it is not, it may not be $k$-dominated by any other tuple in its group.  It is
then a \kdom tuple when only its group is concerned.

Based on the above notion of $k$-dominance within a group, each base relation
$R_i$ is divided into $3$ \emph{mutually exclusive} and \emph{exhaustive} sets
$\ss_i$, $\sn_i$, and $\nn_i$:
\begin{align}
	R_i = \ss_i \cup \sn_i \cup \nn_i
\end{align}

We next define the three sets.

\begin{defn}[\ss]
	A tuple $u$ is in \ss if $u$ is a \kdom skyline in the \emph{overall}
	relation; consequently, it is a \kdom skyline in its group as well.
\end{defn}

\begin{defn}[\sn]
	A tuple $u$ is in \sn if $u$ is a \kdom skyline \emph{only} in its
	\emph{group} but not in the overall relation.
\end{defn}

\begin{defn}[\nn]
	A tuple $u$ is in \nn if $u$ is \emph{not} a \kdom skyline in its group;
	consequently, it is not a \kdom skyline in the overall relation as well.
\end{defn}


\begin{table}[t]
	\centering
	\begin{tabular}{|c|c|cccc|c|}
		\hline
		fno & destination & cost & dur & rtg & amn & category \\
		\hline
		\hline
		11 & C & 448 & 3.2 & 40 & 40 & $\ss_1$ \\
		12 & C & 468 & 4.2 & 50 & 38 & $\nn_1$ \\
		13 & D & 456 & 3.8 & 60 & 34 & $\sn_1$ \\
		14 & D & 460 & 4.0 & 70 & 32 & $\nn_1$ \\
		15 & E & 450 & 3.4 & 30 & 42 & $\sn_1$ \\
		16 & F & 452 & 3.6 & 20 & 36 & $\ss_1$ \\
		17 & G & 472 & 4.6 & 80 & 46 & $\sn_1$ \\
		18 & H & 451 & 3.7 & 20 & 37 & $\ss_1$ \\
		19 & E & 451 & 3.7 & 40 & 37 & $\nn_1$ \\
		\hline
	\end{tabular}
	\tabcaption{Flights from city A ($f_1$).}
	\label{tab:a}
\end{table}

\begin{table}[t]
	\centering
	\begin{tabular}{|c|c|cccc|c|}
		\hline
		fno & source & cost & dur & rtg & amn & category \\
		\hline
		\hline
		21 & D & 348 & 2.2 & 40 & 36 & $\ss_2$ \\
		22 & D & 368 & 3.2 & 50 & 34 & $\nn_2$ \\
		23 & C & 356 & 2.8 & 60 & 30 & $\sn_2$ \\
		24 & C & 360 & 3.0 & 70 & 28 & $\nn_2$ \\
		25 & E & 350 & 2.4 & 30 & 38 & $\sn_2$ \\
		26 & F & 352 & 2.6 & 20 & 32 & $\ss_2$ \\
		27 & G & 372 & 3.6 & 80 & 42 & $\sn_2$ \\
		28 & H & 350 & 2.4 & 35 & 37 & $\sn_2$ \\
		\hline
	\end{tabular}
	\tabcaption{Flights to city B ($f_2$).}
	\label{tab:b}
\end{table}

\begin{table*}[t]
	\centering
	\begin{tabular}{|c|c|cccc|cccc|c|c|}
		\hline
		fno & stop-over & f1.cost & f1.dur & f1.rtg & f1.amn &
		f2.cost & f2.dur & f2.rtg & f2.amn & categorization & skyline \\
		\hline
		\hline
		(11,23) & C & 448 & 3.2 & 40 & 40 & 356 & 2.8 & 60 & 30 & $\ss_1 \Join \sn_2$ & yes \\
		(11,24) & C & 448 & 3.2 & 40 & 40 & 360 & 3.0 & 70 & 28 & $\ss_1 \Join \nn_2$ & no  \\
		(12,23) & C & 468 & 4.2 & 50 & 38 & 356 & 2.8 & 60 & 30 & $\nn_1 \Join \sn_2$ & no  \\
		(12,24) & C & 468 & 4.2 & 50 & 38 & 360 & 3.0 & 70 & 28 & $\nn_1 \Join \nn_2$ & no  \\
		(13,21) & D & 456 & 3.8 & 60 & 34 & 348 & 2.2 & 40 & 36 & $\sn_1 \Join \ss_2$ & yes \\
		(13,22) & D & 456 & 3.8 & 60 & 34 & 368 & 3.2 & 50 & 34 & $\sn_1 \Join \nn_2$ & no  \\
		(14,21) & D & 460 & 4.0 & 70 & 32 & 348 & 2.2 & 40 & 36 & $\nn_1 \Join \ss_2$ & no  \\
		(14,22) & D & 460 & 4.0 & 70 & 32 & 368 & 3.2 & 50 & 34 & $\nn_1 \Join \nn_2$ & no  \\
		(15,25) & E & 450 & 3.4 & 30 & 42 & 350 & 2.4 & 30 & 38 & $\sn_1 \Join \sn_2$ & yes \\
		(16,26) & F & 452 & 3.6 & 20 & 36 & 352 & 2.6 & 20 & 32 & $\ss_1 \Join \ss_2$ & yes \\
		(17,27) & G & 472 & 4.6 & 80 & 46 & 372 & 3.6 & 80 & 42 & $\sn_1 \Join \sn_2$ & no  \\
		(18,28) & H & 451 & 3.7 & 20 & 37 & 350 & 2.4 & 35 & 39 & $\ss_1 \Join \sn_2$ & no  \\
		(19,25) & E & 451 & 3.7 & 40 & 37 & 350 & 2.4 & 30 & 38 & $\nn_1 \Join \sn_2$ & no  \\
		\hline
	\end{tabular}
	\tabcaption{Joined relation ($f_1 \Join f_2$).}
	\label{tab:join}
\end{table*}

Consider the examples in Table~\ref{tab:a} and Table~\ref{tab:b}.  We assume
that all the attributes have lower preferences\footnote{Although the preferences
for ratings and amenities are generally the other way round, we stick to lower
preferences for ease of understanding.}.  We set $k = 3$ for both the relations.
The categorization of the tuples are shown in the last column.

Table~\ref{tab:join} shows the joined relation.

The important outcome of this division into the three sets is that it allows
certain tuples to be automatically designated as \kdom skylines (or not)
\emph{without} computing the join, as explained next.

\subsection{Deciding about Skylines before Join}
\label{sec:sep}

\comment{

For the time being, for simplicity, we assume that the following \emph{unique
value property} (UVP) holds.

\begin{defn}[Unique Value Property, UVP]
	A relation $R$ holds the \emph{unique value property (UVP)} if and only if
	for each (skyline) attribute of $R$, all the tuples have distinct values,
	i.e., all the tuples are unique.
\end{defn}

While the situation is quite common for real-valued attributes, the base
relations may not follow the assumption when there are categorical or
integer-valued attributes.  In Sec.~\ref{sec:nouvp}, we remove the UVP
assumption.

Since each relation is broken into $3$ groups, the join of the two base
relations can be represented as a union of join of $9$ sub-relations.

}


We first analyze a simple case where out of the final $k$ attributes in the
joined relation, if a joined tuple $t$ $k$-dominates another joined tuple $s$,
$t$ must $k_1$-dominate $s$ in the first base relation and must $k_2$-dominate
in the second base relation ($k = k_1 + k_2$).

Note that, in practice this situation will not be specified by any user.  We
are only going to use it to explain the basic concepts and then build upon it
later by removing this assumption.  For the theorems and observations in this
section, the first base relation is divided into the three sets $\ss_1$,
$\sn_1$, $\nn_1$ according to $k_1$-domination while the second base relation
is divided into $\ss_2$, $\sn_2$, $\nn_2$ using $k_2$-domination.

The first theorem shows that a tuple formed by joining the \ss counterparts is
\emph{always} a \kdom skyline.

\begin{thm}
	\label{thm:sepss}
	The tuples in the set $(\ss_1 \Join \ss_2)$ are \kdom skylines.
\end{thm}

\begin{proof}
	Consider a joined tuple $t' = u' \Join v'$ formed by joining the tuples $u'
	\in \ss_1$ and $v' \in \ss_2$.  Assume that $t = u \Join v$ $k$-dominates
	$t'$, i.e., $t \succk t'$.  Since $u' \in \ss_1$, no tuple and, in
	particular, $u$ can $k_1$-dominate $u'$.  Similarly, $v \not\succ_{k_2}
	v'$.  Therefore, $\not\exists t, \ t \succk t'$.  Hence, $t'$ is a \kdom
	tuple.
	\hfill{}
\end{proof}

The flight combination (16,26) in Table~\ref{tab:join} is an example.

The second theorem shows the reverse: composite tuples formed by joining a base
tuple in \nn can \emph{never} be \kdom skylines.

\begin{thm}
	\label{thm:sepnn}
	The tuples in the sets $(\ss_1 \Join \nn_2)$, $(\sn_1 \Join \nn_2)$, $(\nn_1
	\Join \ss_2)$, $(\nn_1 \Join \sn_2)$, and $(\nn_1 \Join \nn_2)$ are not
	\kdom skylines.
\end{thm}

\begin{proof}
	Consider a joined tuple $t' = u' \Join v'$ formed by joining the tuples $u'
	\in \nn_1$ and $v' \in R_2$.  Therefore, there must exist a tuple $u$ in
	the same group that $k_1$-dominates $u'$, i.e., $\exists u, \ u \succ_{k_1}
	u'$.  Consider the tuple $t$ formed by joining $u$ with $v'$.  Since $u$ is
	in the same group as $u'$, this tuple necessarily exists.  As $t$ dominates
	$t'$ in $k_1$ attributes and is equal in $k_2$, overall, it dominates in
	$k_1 + k_2 = k$ attributes, i.e., $t \succk t'$.  Therefore, $t'$ is not a
	\kdom skyline.  This covers the cases $(\nn_1 \Join \ss_2)$, $(\nn_1 \Join
	\sn_2)$, and $(\nn_1 \Join \nn_2)$.  The cases for $(\ss_1 \Join \nn_2)$
	and $(\sn_1 \Join \nn_2)$ are symmetrical.
	\hfill{}
\end{proof}

The flight combinations (11,24), (13,22), (14,21), (14,22), (12,23), and
(12,24) in Table~\ref{tab:join} exemplify the above cases.  For example, the
tuple (11,24) is dominated by (11,23).

However, nothing can be concluded surely about the rest of the joined tuples.

\begin{obs}
	\label{obs:sepssn}
	The tuples in the sets $(\ss_1 \Join \sn_2)$ and $(\sn_1 \Join \ss_2)$ are
	most likely to be \kdom skylines, although that is not guaranteed.
\end{obs}


\begin{proof}
	Consider a joined tuple $t' = u' \Join v'$ formed by joining the tuples $u'
	\in \ss_1$ and $v' \in \sn_2$.  Consider the dominator $v \in R_2
	\succ_{k_2} v'$.  Since $v'$ is a $k_2$-dominate skyline in $R_2$, $v$ is in
	a different group than $v'$.  Therefore, $v$ cannot join with $u'$, i.e.,
	$\not~\!\exists t = u' \Join v, \ v \succ_{k_2} v'$.  Although no tuple $u$
	can $k_1$-dominate $u'$, there may exist $u$ whose $k_1$ attributes have
	the same value as $u'$.  If it happens that $u$ is join-compatible with
	$v$, then and only then, $t = u \Join v$ exists and $t \succ_k t'$ (since
	it is equal in $k_1$ and better in $k_2$ attributes).  The above situation
	is quite unlikely, although not impossible.  The case for $(\sn_1 \Join
	\ss_2)$ is symmetrical.
	\hfill{}
\end{proof}

Thus, while (11,23) and (13,21) are \kdom skylines, (18,28) is not
(Table~\ref{tab:join}).  Flight 18 has same $k_1$ values as 19 while 28 is
dominated by 25.  Therefore, (18,28) is $k$-dominated by (19,25).

The observation for $(\sn_1 \Join \sn_2)$ is similar.

\begin{obs}
	\label{obs:sepsn}
	A tuple in the set $(\sn_1 \Join \sn_2)$ may or may not be a \kdom skyline.
\end{obs}

\begin{proof}
	Consider a joined tuple $t' = u' \Join v' \in (\sn_1 \Join \sn_2)$.
	Consider the tuples $u$ and $v$ such that $u \succ_{k_1} u'$ and $v
	\succ_{k_2} v'$.  Since $u' \in \sn_1$, $u$ is in a different group than
	$u'$.  Similarly, $v$ is in a different group than $v'$.  If and only if
	$u$ and $v$ are join compatible, they will join to form $t = u \Join v$
	which then $k$-dominates $t'$.  Otherwise, no such $t$ exists, and $t'$
	is a \kdom skyline.
	\hfill{}
\end{proof}

Consider (15,25) in Table~\ref{tab:join}.  Since its dominators, flights 11 and
21 respectively, are not join compatible (11 reaches city C while 21 takes off
from city D), the flight combination (11,21) is not valid.  Consequently, no
joined tuple dominates it, and (15,25) becomes a \kdom skyline.  On the other
hand, for (17,27), the dominators 16 and 26 do join (the city F is common) to
form the tuple (16,26).  As a result, (17,27) is not a \kdom skyline.

The overall situation is summed up in Table~\ref{tab:sepsubsets}.

\begin{table}[t]
	\centering
	\begin{tabular}{|c||c|c|c|}
		\hline
		& $\ss_2$ & $\sn_2$ & $\nn_2$ \\
		\hline
		\hline
		$\ss_1$ & yes (Th. \ref{thm:sepss}) & likely (Obs. \ref{obs:sepssn}) & no (Th. \ref{thm:sepnn}) \\
		\hline
		$\sn_1$ & likely (Obs. \ref{obs:sepssn}) & may be (Obs. \ref{obs:sepsn}) & no (Th. \ref{thm:sepnn}) \\
		\hline
		$\nn_1$ & no (Th. \ref{thm:sepnn}) & no (Th. \ref{thm:sepnn}) & no (Th. \ref{thm:sepnn}) \\
		\hline
	\end{tabular}
	\tabcaption{Fate of \kdom skylines.}
	\label{tab:sepsubsets}
\end{table}

\subsection{Skyline Attributes in Joined Relation}
\label{sec:joined}

We next do away with the assumption that $k_1$ and $k_2$ attributes have to be
satisfied separately from the first and second relations.  It is simply
required that the joined relation returns \kdom skylines with no restriction on
how $k$ is broken up between the base relations.  However, to ensure that the
skyline preferences are respected for \emph{at least} one attribute in every
base relation, we assume that $k > \max\{d_1, d_2\}$.

A brute-force way to find the skylines is to generate all combinations of $k_1$
and $k_2$ such that $k_1 + k_2 = k$ and, then, combine the answer sets using
the results obtained in the previous section.  However, it can be done more
efficiently as explained next.

We consider the following cases:
\begin{align}
	\koned = k - d_2 \qquad \ktwod = k - d_1
\end{align}

For all combinations of $k_1$ and $k_2$ such that $k_1 + k_2 = k$, the
inequalities, $1 \leq \koned \leq k_1 \leq d_1$ and $1 \leq \ktwod \leq k_2
\leq d_2$, hold.

Consider the example in Table~\ref{tab:join}.  If $k = 7$, then $\koned = \ktwod
= 3$.  Thus, categorization and \kdom skyline sets remain the same.

We first establish the following lemma on the monotonicity of the \emph{number
of skyline attributes}.

\begin{lem}
	\label{lem:monotonic}
	If tuple $u$ is a $j$-dominant skyline, it is also a $i$-dominant skyline
	for any $i \geq j$.
\end{lem}

\begin{proof}
	Consider $u$ to be a $j$-dominant tuple.  Assume that, however, it is not
	$i$-dominant for some $i > j$.  Thus, there exists $v$ that $i$-dominates
	$u$, i.e., $v$ is better than $u$ in $i$ attributes.  Then, $v$ must be
	better than $u$ in some subset $j$ of these $i$ attributes, which is a
	contradiction.  Thus, $u$ must be $i$-dominant for every $i \geq j$.
	\hfill{}
\end{proof}

The following theorems and observations implicitly use
Lemma~\ref{lem:monotonic}.

\begin{thm}
	\label{thm:ss}
	The tuples in the set $(\ss_1 \Join \ss_2)$ are \kdom skylines.
\end{thm}

\begin{proof}
	Consider $t' = u' \Join v' \in (\ss_1 \Join \ss_2)$.  There are two cases
	to consider: (i)~when there does not exist any tuple $u$ that shares
	$\koned$ attributes with $u'$, and (ii)~when there exists a tuple $u$ that
	share $\koned$ or more attributes with $u'$.
	
	In the first case, if $\exists t, \ t \succk t'$, then $t$ must dominate
	$t'$ in at least $\koned$ attributes corresponding to $u'$ since it can
	dominate $t$ in at most $d_2$ attributes corresponding to $v'$.  Since this
	is impossible, $t'$ is a \kdom skyline.

	In the second case, $u'$ can be dominated by $u \in R_1$ in at most $d_1 -
	1$ attributes (otherwise $u' \not\in \ss_1$).  Thus, to dominate $t'$, $u$
	must combine with $v \in R_2$ such that $v$ dominates $v'$ in at least $k -
	(d_1 - 1) = k_2 + 1$ attributes, which is impossible since $v' \in \ss_2$.
	Hence, $t'$ is a \kdom skyline.
	\hfill{}
\end{proof}

\begin{table}[t]
	\centering
	\begin{tabular}{|c||c|c|c|}
		\hline
		& $\ss_2$ & $\sn_2$ & $\nn_2$ \\
		\hline
		\hline
		$\ss_1$ & yes (Th. \ref{thm:ss}) & likely (Obs. \ref{obs:ssn}) & no (Th. \ref{thm:nn}) \\
		\hline
		$\sn_1$ & likely (Obs. \ref{obs:ssn}) & may be (Obs. \ref{obs:sn}) & no (Th. \ref{thm:nn}) \\
		\hline
		$\nn_1$ & no (Th. \ref{thm:nn}) & no (Th. \ref{thm:nn}) & no (Th. \ref{thm:nn}) \\
		\hline
	\end{tabular}
	\tabcaption{Joins of groups.}
	\label{tab:subsets}
\end{table}

\begin{thm}
	\label{thm:nn}
	The tuples in the sets $(\ss_1 \Join \nn_2)$, $(\sn_1 \Join \nn_2)$, $(\nn_1
	\Join \ss_2)$, $(\nn_1 \Join \sn_2)$, and $(\nn_1 \Join \nn_2)$ are not
	\kdom skylines.
\end{thm}

\begin{proof}
	In each of the sets, at least one base tuple is in \nn.  Consider a joined
	tuple $t' = u' \Join v'$ where $u' \in \nn_1$.  Surely, $\exists u, \ u
	\succ_{\koned} u'$ exists and $u$ is in the same group as $u'$.  Therefore,
	$u$ is join compatible with $v'$.  Consider $t = u \Join v'$.  It dominates
	$t'$ in $\koned + d_2 = k$ attributes.  Thus, $t'$ is not a \kdom skyline.
	\hfill{}
\end{proof}

For example, (16,26) $\in \ss_1 \Join \ss_2$ and (14,22) $\in \nn_1 \Join
\nn_2$ (in Table~\ref{tab:join}) is a \kdom skyline and not a \kdom skyline
respectively.

\begin{obs}
	\label{obs:ssn}
	A tuple in the set $(\ss_1 \Join \sn_2)$ or $(\ss_2 \Join \sn_1)$ is most
	likely to be a \kdom skyline, although that is not guaranteed.
\end{obs}

\begin{proof}
	Consider $t' = u' \Join v' \in (\ss_1 \Join \sn_2)$.  There are two cases
	to consider.  In the first case when there does not exist any tuple $u$
	that share $\koned$ attributes with $u'$, a joined tuple $t$ can dominate
	$t'$ in at most $\koned - 1 + d_2 = k - 1$ attributes.  Therefore, $t'$ is
	a \kdom skyline.

	In the second case, assume that such a $u$ with equal $\koned$ attributes
	exists.  Note, however, that $u$ cannot be better in any other attribute as
	then $u'$ would be $\koned$-dominated by $u$.  Since $v' \in \sn_2$, there
	exists a $v \in R_2$ that dominates $v'$ but is in a different group.  If
	and only if $u$ and $v$ are join compatible, the joined tuple $t = u \Join
	v$ exists.  Comparing $t$ against $t'$, we see that the \koned attributes
	corresponding to $u'$ (or $u$) are same.  Therefore, for $t'$ to
	$k$-dominate $t$, $v'$ must dominate $v$ in all the $k - \koned = d_2$
	attributes.  This is unlikely, although not impossible.

	The proof for $(\ss_2 \Join \sn_1)$ is similar.
	\hfill{}
\end{proof}

For example, consider (18,28) $\in \ss_1 \Join \sn_2$ in Table~\ref{tab:join}
against (19,25).  While $\koned = 3$ attributes for 18 and 19 are same, 25
dominates 28 in $d_2 = 4$ attributes.  Thus, overall, (19,25) dominates (18,28)
in $3 + 4 = 7$ attributes and, thus, (18,28) is not a \kdom skyline.  On the
other hand, (11,23) $\in \ss_1 \Join \sn_2$ is a \kdom skyline since there is
no tuple that shares $\koned = 3$ attributes with 11.

\begin{obs}
	\label{obs:sn}
	A tuple in the set $(\sn_1 \Join \sn_2)$ may or may not be a \kdom skyline.
\end{obs}

\begin{proof}
	Consider a joined tuple $t' = u' \Join v' \in (\sn_1 \Join \sn_2)$.
	Consider the tuples $u$ and $v$ that dominate $u'$ and $v'$ respectively.
	Since $u$ and $v$ are in different groups than $u'$ and $v'$, they may not
	be join compatible.  If they are, the joined tuple $t = u \Join v$
	dominates $t'$ in at least $\koned + \ktwod$ attributes.  The tuple $t$ may
	additionally dominate $t'$ in some other attributes such that it overall
	$k$-dominates $t'$.  If these two conditions are not met, $t'$ becomes a
	\kdom skyline.
	\hfill{}
\end{proof}

Out of the two tuples in $\sn_1 \Join \sn_2$ in table~\ref{tab:join}, while
(15,25) is a \kdom skyline because the dominators 11 and 21 are not in the same
group and, therefore, cannot join, the flight combination (17,27) is not a
\kdom skyline as the dominators 16 and 26 join to form (16,26) which overall
$k$-dominates (17,27).

Table~\ref{tab:subsets} summarizes the situation.

It is important to note that for determining the groups in the base relations,
it must be the minimum number of attributes considered, i.e., $\koned$ and
$\ktwod$; otherwise, the correctness of the above theorems may be violated.

\subsection{Unique Value Property}
\label{sec:uvp}

For a tuple in $\ss_1 \Join \sn_2$ (and $\sn_2 \Join \ss_1$) to be \emph{not} a
\kdom skyline, the number of attributes in which $u \in R_1$ is same as $u' \in
\ss_1$ must be at least $\koned$.  If the base relations follow a \emph{unique
value property} where it is guaranteed that for any $k'_i \ (i = 1,2)$ number
of attributes, two tuples will be unique, then the processing becomes simpler.

\begin{defn}[Unique Value Property]
	A relation $R$ has the \emph{unique value property (UVP)} with respect to
	$i$ if for each subset of $i$ skyline attributes of $R$, all the tuples are
	unique, i.e., no two tuple will have exactly the same values in any
	$i$-sized subset of attributes.
\end{defn}

If $k'_i = 1$, every tuple for any attribute must be unique.  While real-valued
attributes generally follow that, categorical attributes do not.  However, for
reasonable values of \koned and \ktwod, real datasets having a mix of both
real-valued and categorical attributes generally follow the UVP.

The UVP is extremely useful since it ensures that the tuples in $\ss_1 \Join
\sn_2$ and $\sn_2 \Join \ss_1$ become \kdom skylines as shown next.

\begin{thm}
	\label{thm:uvpssn}
	If relations $R_1$ and $R_2$ follow UVP with respect to \koned and \ktwod
	attributes respectively, the tuples in the sets $(\ss_1 \Join \sn_2)$ and
	$(\sn_1 \Join \ss_2)$ are $k$-dominant skylines.
\end{thm}

\begin{proof}
	Consider a joined tuple $t$ in either of the two sets.  The generic
	situation, as shown in Obs.~\ref{obs:ssn} has two cases.  While the first
	case makes $t$ a \kdom skyline, the UVP precludes the second case.  Thus,
	$t$ is always a \kdom skyline.
	\hfill{}
\end{proof}

Although we present Th.~\ref{thm:uvpssn} for the sake of completeness, we do
not assume it in our experiments (Sec.~\ref{sec:exp}).

\begin{table*}[t]
	\centering
	\begin{tabular}{|c|c|c|cccccc|c|c|}
		\hline
		fno & stop-over & cost & f1.dur & f1.rtg & f1.amn
		& f2.dur & f2.rtg & f2.amn & categorization & skyline \\
		\hline
		\hline
		(11,23) & C & 804 & 3.2 & 40 & 40 & 2.8 & 60 & 30 & $\sn_1 \Join \sn_2$ & yes \\
		(11,24) & C & 808 & 3.2 & 40 & 40 & 3.0 & 70 & 28 & $\sn_1 \Join \nn_2$ & no  \\
		(12,23) & C & 824 & 4.2 & 50 & 38 & 2.8 & 60 & 30 & $\nn_1 \Join \sn_2$ & no  \\
		(12,24) & C & 828 & 4.2 & 50 & 38 & 3.0 & 70 & 28 & $\nn_1 \Join \nn_2$ & no  \\
		(13,21) & D & 804 & 3.8 & 60 & 34 & 2.2 & 40 & 36 & $\sn_1 \Join \sn_2$ & yes \\
		(13,22) & D & 824 & 3.8 & 60 & 34 & 3.2 & 50 & 34 & $\sn_1 \Join \nn_2$ & no  \\
		(14,21) & D & 808 & 4.0 & 70 & 32 & 2.2 & 40 & 36 & $\nn_1 \Join \sn_2$ & no  \\
		(14,22) & D & 828 & 4.0 & 70 & 32 & 3.2 & 50 & 34 & $\nn_1 \Join \nn_2$ & no  \\
		(15,25) & E & 800 & 3.4 & 30 & 42 & 2.4 & 30 & 38 & $\sn_1 \Join \sn_2$ & yes \\
		(16,26) & F & 804 & 3.6 & 20 & 36 & 2.6 & 20 & 32 & $\ss_1 \Join \ss_2$ & yes \\
		(17,27) & G & 844 & 4.6 & 80 & 46 & 3.6 & 80 & 42 & $\sn_1 \Join \sn_2$ & no  \\
		(18,28) & H & 801 & 3.7 & 20 & 37 & 2.4 & 35 & 39 & $\ss_1 \Join \sn_2$ & no  \\
		(19,25) & E & 801 & 3.7 & 40 & 37 & 2.4 & 30 & 38 & $\nn_1 \Join \sn_2$ & no  \\
		\hline
	\end{tabular}
	\tabcaption{Joined relation ($f_1 \Join f_2$): aggregate.}
	\label{tab:joinagg}
\end{table*}

\subsection{Aggregate Attributes}
\label{sec:aggregate}

We next consider the case of aggregation where the \kdom skyline is sought over
attributes that attain values aggregated from attributes in the base relation.

We assume that there are $l_1$ local and $a$ aggregate skyline attributes in
$R_1$ and $l_2$ local and $a$ aggregate skyline attributes in $R_2$.  The total
number of skyline attributes in $R$ is, therefore, $l_1 + l_2 + a$ attributes.
As earlier, the final skyline query is asked over $k < l_1 + l_2 + a$
attributes.

The groups in the base relations are partitioned based on both the local and
aggregate attributes.  Out of the final number of skyline attributes $k$, $a$ of
them can be aggregate.  Thus, the minimum number of local attributes that must
be dominated in each base relation is $\konedd = k - a - l_2$ and $\ktwodd = k -
a - l_1$.  The categorization of the base relations into the sets \ss, \sn and
\nn are done on the basis of $\koned = \konedd + a$ and $\ktwod = \ktwodd + a$.
Since $d_1 = a + l_1$ and $d_2 = a + l_2$, these definitions are
same as earlier (Sec.~\ref{sec:joined}).

We use the following assumption about the \emph{monotonicity} property of
the aggregate attributes.

\begin{assmp}[Monotonicity]
	If the value of attribute $u_1$ dominates that of $u_2$ and the value of
	attribute $v_1$ dominates that of $v_2$, the aggregated value of $u_1
	\oplus v_1$ will dominate the aggregated value of $u_2 \oplus v_2$, where
	$\oplus$ denotes the aggregation operator.
\end{assmp}

Since the categorization remains the same, the fate of the joined tuples using
the aggregation remains exactly the same as earlier (as summarized in
Table~\ref{tab:subsets}).

Table~\ref{tab:joinagg} shows the joined relation obtained from
Table~\ref{tab:a} and Table~\ref{tab:b} with the cost values aggregated.

In the example, considering $k = 6$ with $a = 1$, $\konedd = 6 - 1 - 3 = 2$ and
$\ktwodd = 6 - 1 - 3 = 2$ but $\koned = \konedd + 1 = 3$ and $\ktwod = \ktwodd
+ 1 = 3$ remain as earlier.

\section{Algorithms}
\label{sec:algo}

In this section, we describe the various algorithms for answering the \kdom
skyline join queries (Sec.~\ref{sec:joined}).  We consider the general case
where the datasets do not follow UVP (Sec.~\ref{sec:uvp}) and, where in
addition to local skyline attributes, there are aggregate ones as well
(Sec.~\ref{sec:aggregate}).

\subsection{Na\"ive Algorithm}

The na\"ive algorithm (Algo.~\ref{alg:ksjqnaive}) simply computes the join of
the two relations first (line 1) and then computes the \kdom skylines from the
joined relation (line 2) using any of the standard \kdom skyline computation
methods \cite{skylinekdominant}.  Being the most basic, it suffers from two
major disadvantages.  First, the join can require a very large time, thereby
rendering the entire algorithm extremely time-consuming and impractical.  The
second shortcoming is the non-progressive result generation.  The user has to
wait a fairly large time (at least the complete joining time) before even the
first skyline result is presented to her.  In online scenarios, the progressive
result generation is quite an attractive and useful feature.

\begin{algorithm}[t]\footnotesize
	\caption{KSJQ: Na\"ive Algorithm}
	\label{alg:ksjqnaive}
	\begin{algorithmic}[1]
		\Require Relations $R_1, R_2$; Number of attributes $k$
		\Ensure \kdom skyline set $T$
		\State $D \gets R_1 \Join R_2$
		\State $T \gets$ \kdom skyline$(D, k)$
		\State \Return $T$
	\end{algorithmic}
\end{algorithm}

\subsection{Target Set}
\label{sec:targetset}

To alleviate the problems of the na\"ive algorithms, we next propose two
algorithms, \emph{grouping} and \emph{dominator-based}, that use the concepts
of optimization from Sec.~\ref{sec:opt}.

However, before we describe them, we first explain and define the concept of
\emph{target sets}.  Although a tuple in a set marked by ``may be'' or
``likely'' is not guaranteed to be a \kdom skyline, it needs to be checked
against only a small set of tuples, called its \emph{target set}.

Formally, a target set for a tuple $u'$ in a base relation is the set of tuples
$\ts(u')$ that can potentially combine with other tuples from the other base
relation and $k$-dominate a joined tuple formed with $u'$.  In other words, for
a joined tuple $t' = u' \Join v'$, there \emph{may} exist $v$ such that $t = u
\Join v$ where $u \in \ts(u')$ $k$-dominates $t'$.  No tuple outside the target
set $\ts(u)$ of $u$ may combine with any other tuple and dominate $t$.

\begin{defn}[Target Set]
	The target set for a tuple $u' \in R_i$ is the set of tuples $\ts(u')
	\subseteq R_i$ such that
	$\forall u \not\in \ts(u'), \ \not\exists v,v', \ t = u \Join v \succ t' =
	u' \Join v'$.
\end{defn}

The definition is one-sided: a tuple in the target set may or may not
join and dominate $t'$, but no tuple outside the target set can join and
dominate $t'$.  The utility of a target set is easy to understand.  To check
whether $t'$ is a \kdom skyline, $u'$ needs to be checked only against its
target set and nothing outside it.

The \emph{join} of target sets for the base tuples produces the potential
dominating set for the joined tuple.

The target set for a tuple $u' \in \ss$ constitutes itself and the set of
tuples $\{u\}$ that has \emph{at least} $\koned$ attributes same as $u'$.  The
augmentation is required to guarantee the correctness as explained in
Obs.~\ref{obs:ssn}.  The tuple $u'$ must be included in the target set of $u'$
for the same reason.

For each tuple in \ss, the number of such tuples sharing at least \koned
attributes is typically low.  Hence, maintenance of target sets is quite
feasible and practical.

However, the target set for a tuple in \sn can be any tuple \emph{outside} its
group that dominates it.  For simplicity, we consider it as the entire dataset
$R_i$.

Similarly, the target set for \nn is $R_i$.

\subsection{Grouping Algorithm}

Our first algorithm, called the \emph{grouping} algorithm
(Algo.~\ref{alg:ksjqgrouping}), first computes the groups \ss, \sn and \nn in
the base relations.  Next, the summarization from Table~\ref{tab:subsets} is
used.  Tuples from the sets marked by ``yes'' are immediately output as \kdom
skyline tuples.  Tuples marked by ``no'' are pruned and not even joined.

A tuple in a set marked by ``may be'' or ``likely'' is checked against the join
of the \emph{target} sets of the base tuples.

For a base tuple in the \ss group, the target set is first augmented with
tuples that share at least $k'_i$ attributes (the Augment subroutine in
lines~\ref{line:augmentone} and~\ref{line:augmenttwo}).  Then, each group of
tuples is checked only against its target set.  For example, in
line~\ref{line:targetone}, the set $\ss_1 \Join \sn_2$ is checked only against
$A_1 \Join R_2$ for a domination in $k$ attributes.  Note that the target set
for $u \in \ss_1$ is only $A_1$.  Similarly, lines~\ref{line:targettwo}
and~\ref{line:targetthree} handle the sets $\sn_1 \Join \ss_2$ and $\sn_2 \Join
\sn_1$ respectively.


The efficiency of the grouping algorithm stems from the fact that only the
tuples in $\sn_1 \Join \sn_2$ need to be compared against the entire $R_1 \Join
R_2$.  For other tuples, the decision can be taken without even joining (the
``yes'' and ``no'' cases), or the comparison set is small (for the tuples in
$\ss_1 \Join \sn_2$ and $\ss_2 \Join \sn_1$).

\begin{algorithm}[t]\footnotesize
	\caption{KSJQ: Grouping Algorithm}
	\label{alg:ksjqgrouping}
	\begin{algorithmic}[1]
		\Require Relations $R_1, R_2$; Number of attributes $k$
		\Ensure \kdom skyline set $T$
		\State $\koned \gets k - d_2$
		\State $\ktwod \gets k - d_1$
		\State $\ss_1, \sn_1, \nn_1 \gets$ Group$(R_1, \koned)$
		\State $\ss_2, \sn_2, \nn_2 \gets$ Group$(R_2, \ktwod)$
		\State $T \gets$ $(\ss_1 \Join \ss_2)$ \Comment {``yes'' tuples}
		\State $A_1 \gets$ Augment$(\ss_1, \koned)$
			\label{line:augmentone}
			\Comment {augment $u \in \ss_1$ with $\{u': u'_{\koned} = u_{\koned}\}$}
		\State $A_2 \gets$ Augment$(\ss_2, \ktwod)$
			\label{line:augmenttwo}
			\Comment {augment $v \in \ss_2$ with $\{v': v'_{\ktwod} = v_{\ktwod}\}$}
		\State $T_1 \gets$ CheckTarget$(\ss_1 \Join \sn_2, A_1 \Join R_2, k)$
			\label{line:targetone}
		\State $T_2 \gets$ CheckTarget$(\sn_1 \Join \ss_2, R_1 \Join A_2, k)$
			\label{line:targettwo}
		\State $T_3 \gets$ CheckTarget$(\sn_1 \Join \sn_2, R_1 \Join R_2, k)$
			\label{line:targetthree}
		\State \Return $T \gets T_1 \cup T_2 \cup T_3$
	\end{algorithmic}
\end{algorithm}

\subsection{Dominator-Based Algorithm}

The grouping algorithm has the problem that for tuples in $\sn_i$, the target
set is the entire relation $R_i$.  The next algorithm, \emph{dominator-based}
algorithm (Algo.~\ref{alg:ksjqdominator}) rectifies this by \emph{explicitly}
storing the set of dominators.  Thus, for a tuple $\in \ss_i, \sn_i$, the
dominator set of tuples is first obtained (lines~\ref{line:domone}
and~\ref{line:domtwo}).  When the tuple is in $\ss_i$, this set is empty.  The
dominator sets are augmented by the tuples themselves and those tuples having
the same values in the required number of skyline attributes
(lines~\ref{line:augone} and~\ref{line:augtwo}).  In general, the size of the
dominator sets is quite low as compared to the entire dataset, i.e., generally
$|dom(\cdot)| \ll R_i$.

The target sets for the joined tuples are composed of the joins of the
dominating tuples.  Each tuple in the sets marked by ``likely'' and ``may be''
is checked against these joins of the corresponding dominator sets and is added
to the answer only if no dominator exists (line~\ref{line:checkdom}).

The joins of dominating sets are substantially less in size than the target
sets for the grouping algorithm (which is the join of the entire target sets).
The saving is largest for tuples in $\sn_1 \Join \sn_2$.

The saving, however, comes at a cost.  For each tuple in \sn, \emph{all} the
dominators need to found out.  In addition, the entire dominator set needs to be
stored explicitly.

The advantages of the dominator-based algorithm may not be enough to offset
this overhead of time and storage, especially when there are many such tuples.
Sec.~\ref{sec:exp} compares the different algorithms empirically.

\begin{algorithm}[t]\footnotesize
	\caption{KSJQ: Dominator-Based Algorithm}
	\label{alg:ksjqdominator}
	\begin{algorithmic}[1]
		\Require Relations $R_1, R_2$; Number of attributes $k$
		\Ensure \kdom skyline set $T$
		\State $\koned \gets k - d_2$
		\State $\ktwod \gets k - d_1$
		\State $\ss_1, \sn_1, \nn_1 \gets$ Group$(R_1, \koned)$
		\State $\ss_2, \sn_2, \nn_2 \gets$ Group$(R_2, \ktwod)$
		\State $T \gets$ $(\ss_1 \Join \ss_2)$ \Comment {``yes'' tuples}
		\ForAll {$u \in \ss_1, \sn_1$}
			\State $dom(u) \gets$ $\koned$-dominators$(u)$ 
				\label{line:domone}
				\Comment {dominators of $u$ with $\koned$ attributes}
			\State $dom(u) \gets dom(u) \cup$ Augment$(u, \koned)$
				\label{line:augone}
				\Comment {$\{u': u'_{\koned} = u_{\koned}\}$}
		\EndFor
		\ForAll {$v \in \ss_2, \sn_2$}
			\State $dom(v) \gets$ $\ktwod$-dominators$(v)$ 
				\label{line:domtwo}
				\Comment {dominators of $v$ with $\ktwod$ attributes}
			\State $dom(v) \gets dom(v) \cup$ Augment$(v, \ktwod)$
				\label{line:augtwo}
				\Comment {$\{v': v'_{\ktwod} = v_{\ktwod}\}$}
		\EndFor
		\State $T \gets \varnothing$
		\ForAll {$u \Join v \in (\ss_1 \Join \sn_2) \cup (\sn_1 \Join \ss_2) \cup (\sn_1 \Join \sn_2)$}
			\State $T \gets T \cup$ CheckDominators($dom(u) \Join dom(v), k$)
				\label{line:checkdom}
		\EndFor
		\State \Return $T$
	\end{algorithmic}
\end{algorithm}

\subsection{Cartesian Product}
\label{sec:cartesian}

When the final relation is a \emph{Cartesian product} of the two base relations,
the algorithms become considerably easier.  The Cartesian product can be
considered as a special case of join with every tuple having the \emph{same}
value of the join attribute.  In other words, all the tuples are in the same
join group.  As a result, there is no \sn set.  A tuple is either in \ss (when
it is a skyline in its local relation) or in \nn (when it is not a skyline).
Consequently, the tables become much simpler, and the fate of all the joined
(i.e., final) tuples can be concluded without the need to explicitly compute
them.  The tuples in $\ss_1 \Join \ss_2$ are skylines while none of the other
tuples are.

\subsection{Non-Equality Join Condition}
\label{sec:nonequality}

In certain cases, the join condition may \emph{not} be an equality.  For
example, in a flight combination, the arrival time of the first leg needs to be
earlier than the departure time of the second, i.e., $f_1.\text{arrival} <
f_2.\text{departure}$.  In this section, we discuss how to handle such join
cases when the condition is one of $<, \leq, >, \geq$.

Since the main purpose of dividing a base relation into the three sets \ss, \sn
and \nn is to ensure that certain decisions can be taken about the tuples in
these sets \emph{without} joining, all the optimizations discussed in
Sec.~\ref{sec:opt} work with the following modifications.

A tuple in $\ss_1 \Join \ss_2$ can never be $k$-dominated and, thus, it does
not matter how such a tuple is composed from the base relations.  In other
words, the semantics of the join condition, equality or otherwise, does not
matter for $\ss_1 \Join \ss_2$ tuples.

Next consider a tuple $t' = u' \Join v' \in (\sn_1 \Join \ss_2)$.  A tuple in
the \sn set, $u'$, is originally defined as one that is not dominated by any
other tuple in the \emph{same} group.  This ensures that if $u'$ joins with
$v'$ from the other relation, then \emph{no} other tuple $u$ can join with $v'$
to dominate $t'$.  This is the crucial property that needs to be maintained
even when the join condition is non-equality.

Thus, if the join condition is $u'.arr < v'.dep$, then the set $\{u: u.arr <
u'.arr\}$ is considered to be in the \emph{same} group of $u'$ since it can
also join with $v'$ (and can potentially dominate $u$).  In other words, this
ensures that all such $u$ is join compatible with $v'$.  The set \sn is thus
expanded to take care of the non-equality condition.  Note that there may exist
other tuples $\{u'': u''.arr < v'.dep\}$ that may also join with $v'$; these,
however, cannot be determined locally without the knowledge of $v'$ from the
other relation and, therefore, cannot be considered.  The target set of a tuple
in \sn is anyway the entire dataset (or its dominators).  Thus, the algorithms
will work correctly with the above modification.

Similarly, for the converse set, $\ss_1 \Join \sn_2$, the \sn set for $v'$
consists of $\{v: v.dep > v'.dep\}$.

A joined tuple with \nn as a component is rejected as a \kdom
skyline since the \nn tuple can be always dominated by another tuple in the
\emph{same} group.  Hence, similar to \sn, the group of a tuple needs to be
defined by taking into account the semantics of the join condition.

Considering the earlier example of $u'.arr < v'.dep$, if $u' \in \nn$, then to
say that $u' \in \nn_1$, there must exist $u: u.arr < u'.arr$ and $u
\succ_{\koned} u'$.  (The definition is suitably modified for $\nn_2$ in a
suitable manner.) This ensures that the joined tuple $t' = u' \Join v'$ will be
dominated by $u \Join v'$.  Once more, tuples of the form $\{u'': u''.arr <
v'.dep\}$ are left out due to lack of knowledge about $v'$.

It may happen that there does not exist any such $u$ but there exists such an
$u''$.  In that case, $u'$ is classified as an \sn tuple instead of an \nn.
This, however, only leads to extra processing of tuples joined with $u'$.  The
\emph{correctness} is not violated as such joined tuples of the form $t' = u'
\Join v'$ will be finally caught by $u'' \Join v'$ and rejected.  Thus, the
above modification only affects the efficiency of the algorithms, not the final
result.

\comment{

We need an experimental evaluation here.

}

\subsection{Algorithms for Aggregation}
\label{sec:algoagg}

The algorithms that consider aggregation
are essentially the same as the plain KSJQ.
The only difference is that when the join is performed, the aggregation of the
attributes are done as an additional step.  Therefore, they are not discussed
separately.

We next discuss the algorithms for finding $k$
(Problem~\ref{prob:k}).

\subsection{Finding k: Na\"ive Algorithm}

The na\"ive algorithm (Algo.~\ref{alg:knaive}) to search for the lowest $k$ that
produces at least $\delta$ skylines starts from the least possible value of $k$,
i.e., $\max\{d_1,d_2\} + 1$, and keeps incrementing it till the number of \kdom
skylines is at least $\delta$.  The largest possible value of $k$, i.e., $d$, is
otherwise returned by default, even if it does not satisfy the $\delta$
criterion.

\begin{algorithm}[t]\footnotesize
	\caption{Finding $k$: Na\"ive Algorithm}
	\label{alg:knaive}
	\begin{algorithmic}[1]
		\Require Number of skylines $\delta$
		\Ensure Number of attributes $k$
		\State $k \gets \max\{d_1,d_2\} + 1$ \Comment {minimum $k$}
		\While {$k < d$}
			\If {$|skyline(k)| \geq \delta$} \Comment {actual number}
				\State \Return $k$ \Comment {return and terminate}
			\EndIf
			\State $k \gets k + 1$
		\EndWhile
		\State \Return $d$ \Comment {maximum possible $k$}
	\end{algorithmic}
\end{algorithm}

The algorithm is very inefficient as for each case, it computes the actual
\kdom skyline set.  Further, it traverses the possibilities of $k$ in a linear
manner.  The correctness is based on the fact that the number of \kdom skylines
is a monotonically non-decreasing function in $k$ (Lemma~\ref{lem:monotonic}).

We next design algorithms that use Table~\ref{tab:subsets}.

\subsection{Finding k: Range-Based Algorithm}

For a particular value of $k$, the actual number of \kdom skylines, $\Delta_k$,
is \emph{at least} the size of the ``yes'' sets, denoted by $\Delta_{k,lb}$,
and \emph{at most} the sum of sizes of the ``yes'', ``likely'' and ``may be''
sets, denoted by $\Delta_{k,ub}$.  These, thus, denote the lower and upper
bounds respectively: $\Delta_{k,lb} < \Delta_k < \Delta_{k,ub}$.

The \emph{range-based} algorithm (Algo.~\ref{alg:krange}) uses these bounds to
speed-up the process.  Starting from the minimum possible $k = \max\{d_1,d_2\}
+ 1$, the algorithm finds $\Delta_{k,lb}$ and $\Delta_{k,ub}$
(lines~\ref{line:delone} and~\ref{line:deltwo} respectively).  If
$\Delta_{k,lb} \geq \delta$, then the current $k$ is the answer
(line~\ref{line:answerone}).  If $\Delta_{k,ub} < \delta$, then the current
$k$ cannot be the answer and $k$ is incremented (line~\ref{line:kplusone}).
Otherwise, i.e., if $\Delta_{k,lb} < \delta \leq \Delta_{k,ub}$, then the
current $k$ may be an answer.  In this case, the \emph{actual} \kdom skyline
set is computed.  If its size is $\delta$ or greater, it is returned as the
answer (line~\ref{line:answertwo}).  Else, $k$ is incremented
(line~\ref{line:kplustwo}), and the steps are repeated.


\begin{algorithm}[t]\footnotesize
	\caption{Finding $k$: Range-Based Algorithm}
	\label{alg:krange}
	\begin{algorithmic}[1]
		\Require Number of skylines $\delta$
		\Ensure Number of attributes $k$
		\State $k \gets \max\{d_1,d_2\} + 1$ \Comment {minimum $k$}
		\While {$k < d$}
			\State $\Delta_{k,lb} \gets |\textrm{``yes'' sets}|$
				\label{line:delone}
			\State $\Delta_{k,ub} \gets |\textrm{``yes'' sets}| + |\textrm{``likely'' sets}| + |\textrm{``may be'' sets}|$
				\label{line:deltwo}
			\If {$\Delta_{k,lb} \geq \delta$} \Comment {lower bound}
				\label{line:answerone}
				\State \Return $k$
			\ElsIf {$\Delta_{k,ub} < \delta$} \Comment {upper bound}
				\label{line:kplusone}
				\State $k \gets k + 1$
			\ElsIf {$|skyline(k)| \geq \delta$} \Comment {actual number}
				\label{line:answertwo}
				\State \Return $k$
			\Else
				\label{line:kplustwo}
				\State $k \gets k + 1$
			\EndIf
		\EndWhile
		\State \Return $d$ \Comment {maximum possible $k$}
			\label{line:return}
	\end{algorithmic}
\end{algorithm}

While this algorithm is definitely more efficient than the na\"ive one, it
still suffers from the shortcoming that it examines a large number of $k$'s by
incrementing it one by one.  If the required $k$ lies towards the higher end of
the range, it unnecessarily examines too many lower values of $k$.  The next
algorithm does a \emph{binary search} to reduce this overhead.

\subsection{Finding k: Binary Search Algorithm}

Algo.~\ref{alg:kbinary} shows how the binary search proceeds.  It starts from
the middle of the possible values of $k$ (line~\ref{line:middle}).  The values
of $\Delta_{k,lb}$ and $\Delta_{k,ub}$ are computed in the same manner as
earlier (lines~\ref{line:bdelone} and~\ref{line:bdeltwo} respectively).

If $\Delta_{k,lb} \geq \delta$ (line~\ref{line:onegeq}), then $k$ is a
potential answer.  No value in the higher range can be the answer as the
current $k$ already satisfies the condition.  However, there may be a lower $k$
that satisfies the $\delta$ condition.  Hence, the search is continued in the
lower range to try and find a better (i.e., lesser) $k$.

If, on the other hand, $\Delta_{k,ub} < \delta$ (line~\ref{line:twoless}), then
the current estimate of $k$ is too low.  The search is, therefore, continued in
the higher range.

If none of these bounds help, the actual number of \kdom skylines, $\Delta_k$,
is found.  If $\Delta_k \geq \delta$ (line~\ref{line:skygeq}), then the current
$k$ is a potential answer.  However, a lesser $k$ may be found and, so, the
search proceeds to the lower range.

Otherwise, i.e., when $\Delta_k < \delta$ (line~\ref{line:skyless}), the
required $k$ is searched in the higher range.

The algorithm stops when the lower range of the search becomes larger than or
equal to the current answer (line~\ref{line:banswer}), which then is the lowest
$k$ that satisfies the $\delta$ condition.  It may also stop when the range is
exhausted, in which case, the current value of $k$ is returned
(line~\ref{line:breturn}).

\begin{algorithm}[t]\footnotesize
	\caption{Finding $k$: Binary Search Algorithm}
	\label{alg:kbinary}
	\begin{algorithmic}[1]
		\Require Number of skylines $\delta$
		\Ensure Number of attributes $k$
		\State $l \gets \max\{d_1,d_2\} + 1$ \Comment {minimum $k$}
		\State $h \gets d$ \Comment {maximum $k$}
		\State $cur \gets d$ \Comment {current estimate of $k$}
		\While {$l < h$}
			\State $k \gets \lfloor (l+h) / 2 \rfloor$
				\label{line:middle}
			\State $\Delta_{k,lb} \gets |\textrm{``yes'' sets}|$
				\label{line:bdelone}
			\State $\Delta_{k,ub} \gets |\textrm{``yes'' sets}| + |\textrm{``likely'' sets}| + |\textrm{``may be'' sets}|$
				\label{line:bdeltwo}
			\If {$\Delta_{k,lb} \geq \delta$}
				\label{line:onegeq}
				\State $cur \gets k$ \Comment {update current estimate}
				\State $h \gets k - 1$ \Comment {search for a lower $k$}
			\ElsIf {$\Delta_{k,ub} < \delta$}
				\label{line:twoless}
				\State $l \gets k + 1$ \Comment {search for a higher $k$}
			\ElsIf {$|skyline(k)| \geq \delta$}
				\label{line:skygeq}
				\State $cur \gets k$ \Comment {update current estimate}
				\State $h \gets k - 1$ \Comment {search for a lower $k$}
			\ElsIf {$|skyline(k)| < \delta$}
				\label{line:skyless}
				\State $l \gets k + 1$ \Comment {search for a higher $k$}
			\EndIf
			\If {$l \geq cur$} \Comment {lowest $k$ already found}
				\label{line:banswer}
				\State \Return $cur$
			\EndIf
		\EndWhile
		\State \Return $cur$
			\label{line:breturn}
	\end{algorithmic}
\end{algorithm}

The binary search based algorithm, thus, speeds up the searching through the
possible range of values of $k$.  Sec.~\ref{sec:exp} compares the three
algorithms empirically.

\section{Experimental Results}
\label{sec:exp}

We experimented with data synthetically generated using
\url{http://randdataset.projects.pgfoundry.org/} on an Intel i7-4770 @3.40\,GHz
Octacore machine with 16\,GB RAM using code written in Java.

We also experimented with a real dataset of two-legged flights from New Delhi
to Mumbai.  The details are in Sec.~\ref{sec:real}.

We measured the effects of various parameters on the different algorithms
proposed in Sec.~\ref{sec:algo}.  The parameters and their default values are
listed in Table~\ref{tab:param}.  Note that the size of the joined relation is a
derived parameter.  It is equal to $n^2/g$ for two base relations with $n$
tuples and $g$ groups.  When the effect of a particular set of parameters are
measured, the rest are held to their default values, unless explicitly stated
otherwise.

\begin{table}[t]
	\centering
	\begin{tabular}{|c|c|c|}
		\hline
		Symbol & Parameter & Default value \\
		\hline
		\hline
		$n$ & Dataset size for base relation & $3,300$ \\
		$d$ & Dimensionality of base relation & $7$ \\
		$k$ & Number of skyline attributes & $11$ \\
		$a$ & Number of aggregate attributes & $2$ \\
		$g$ & Number of join groups & $10$ \\
		$T$ & Dataset type & Independent \\
		$\delta$ & Threshold of skyline size & $10,000$ \\
		\hline
		$N$ & Size of joined relation & $1,089,000$ \\
		\hline
	\end{tabular}
	\tabcaption{Parameters for experiments.}
	\label{tab:param}
\end{table}

In the figures, the three main algorithms for KSJQ are denoted as: $G$ for
grouping, $D$ for dominator-based, and $N$ for na\"ive.  The overall running
time for each algorithm is divided into various components: (i)~time taken for
computing the groups in the base relations, i.e., \ss, \sn, and \nn, (ii)~time
taken for actually joining the tuples from the two base relations that cannot be
pruned, (iii)~time taken for finding the dominators of the tuples, and (iv)~the
rest of the processing.  These are marked separately in the figures.

Not all the components are applicable to every algorithm, e.g., the na\"ive
algorithm does not find groups.  The components that are not applicable to an
algorithm are shown to have zero costs.

The three algorithms for determining the value of $k$ are depicted as: $B$ for
binary search, $R$ for range-based, and $N$ for na\"ive.

\subsection{Aggregate}

We first show the results where aggregate values have been used.  The
aggregation function used is \emph{sum}.

\comment{

\begin{verbatim}
A
D6_A1
D7_A2
Dimension_2
Group_Size
N
Data_Type
\end{verbatim}

}

\subsubsection{Effect of Dimensionality}

The first experiment measures the effect of varying $k$.  Fig.~\ref{fig:sumk}
shows that the running time increases sharply with $k$.  The two different
settings of dimensionality, $d$, and number of aggregate attributes, $a$,
(Fig.~\ref{subfig:sumk7} and Fig.~\ref{subfig:sumk6}) show the robustness of
this behavior.  As $k$ increases, it becomes increasingly hard to dominate a
tuple in $k$ dimensions.  As a result, the number of \kdom skyline increases
heavily, thereby resulting in increased running times.

Overall, the grouping algorithm is the fastest.  The dominator-based algorithm
spends a substantial amount of time in finding the dominators for each tuple.
This overhead is not compensated enough in the final checking stage.  This is
due to the fact that the average dominator set sizes are quite large and,
hence, joining large dominator sets requires a substantial amount of time.
With increasing dimensionality, the time required to find the dominator sets
increases as well.

As expected, the na\"ive algorithm performs the worst as it does not attempt
any optimization at all.  It is slower than the grouping algorithm by about
1.5-2 times.

\begin{figure}[t]
	\centering
	\subfloat[$d = 7, a = 2$]
	{
		\includegraphics[width=\subfigwidth]{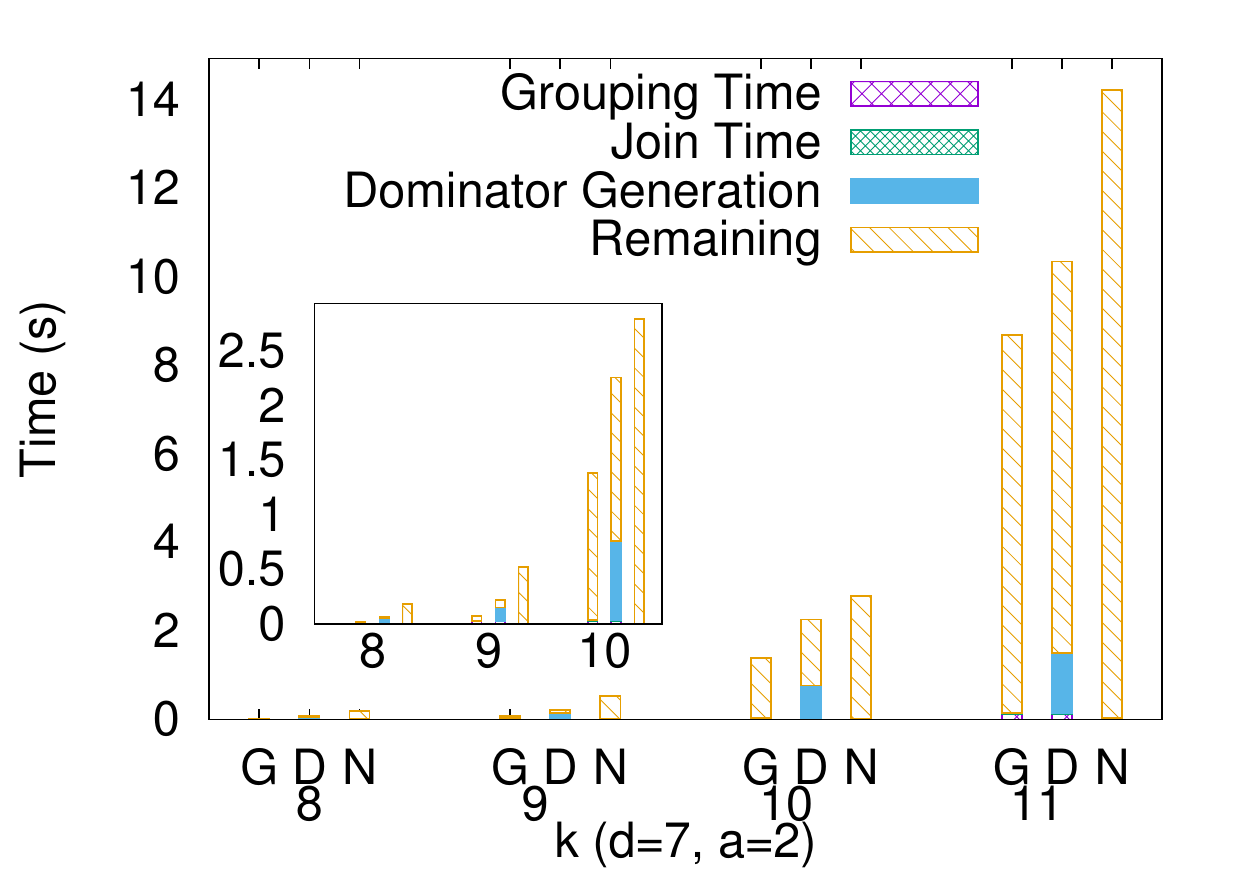}
		\label{subfig:sumk7}
	}
	\subfloat[$d = 6, a = 1$]
	{
		\includegraphics[width=\subfigwidth]{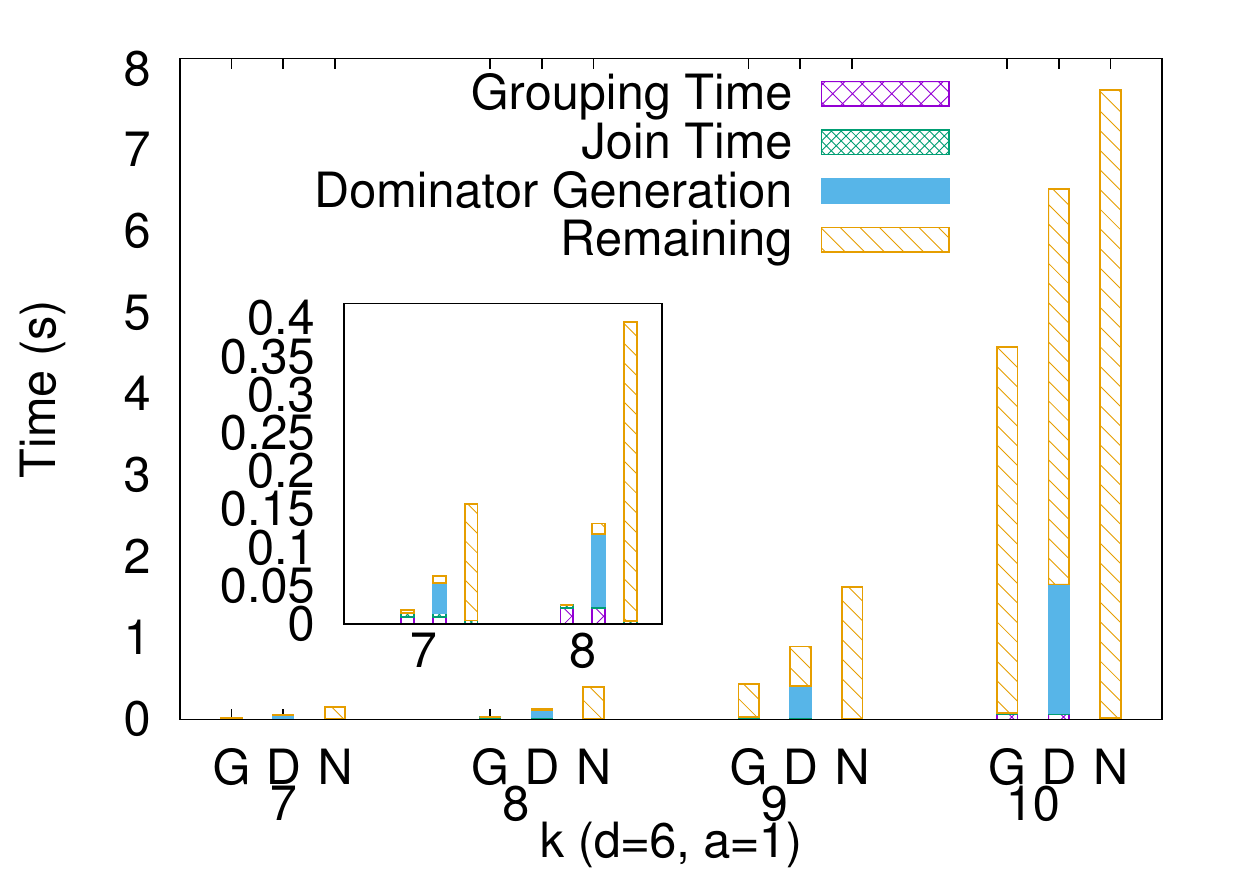}
		\label{subfig:sumk6}
	}
	\figcaption{Effect of $k$.}
	\label{fig:sumk}
\end{figure}

The next set of experiments hold $d$ and $k$ constant but increases the number
of aggregate attributes, $a$.  Fig.~\ref{subfig:suma} depict the results.  Note
that $a = 0$ signifies that no aggregate attributes are used.  The trend of the
results remain the same with the running time increasing with $a$.  Once more,
grouping is the best algorithm followed by dominator-based and na\"ive.

Fig.~\ref{subfig:sumd} shows a medley of results across different $d$, $k$ and
$a$.  When $a$ or $k$ increase, the running time increases as well.

However, it seems that the reverse happens with $d$.  Comparing the case of $d
= 5, k = 7, a = 1$ with $d = 6, k = 7, a = 1$, we note that in the first case,
$\koned = \ktwod = 3$ while in the second case, $\koned = \ktwod = 2$.  Thus,
in the second case, it is easier to find the groups and perform the joins.
Consequently, it runs faster.  The same reasoning holds true for $d = 5, k = 7,
a = 2$ against $d = 6, k = 7, a = 2$.

We next compare $d = 5, k = 7, a = 2$ against $d = 6, k = 8, a = 2$.  The
values of $\koned = \ktwod = 4$ are same in both the cases.  However, the size
of dominator sets is larger in the first case.  The time required to divide the
base relations into the three sets is also higher.  This leads to an overall
higher running time.

\comment{

The final skyline is asked for $k = 7$ attributes out of the total
dimensionality of $8$ attributes in the first case, while in the second case,
it is for $k = 8$ attributes out of a total of $10$ attributes.  Since the
number of tuples is the same in both the cases, the number of skylines in the
first case is higher.

}

\begin{figure}[t]
	\centering
	\subfloat[Effect of $a$.]
	{
		\includegraphics[width=\subfigwidth]{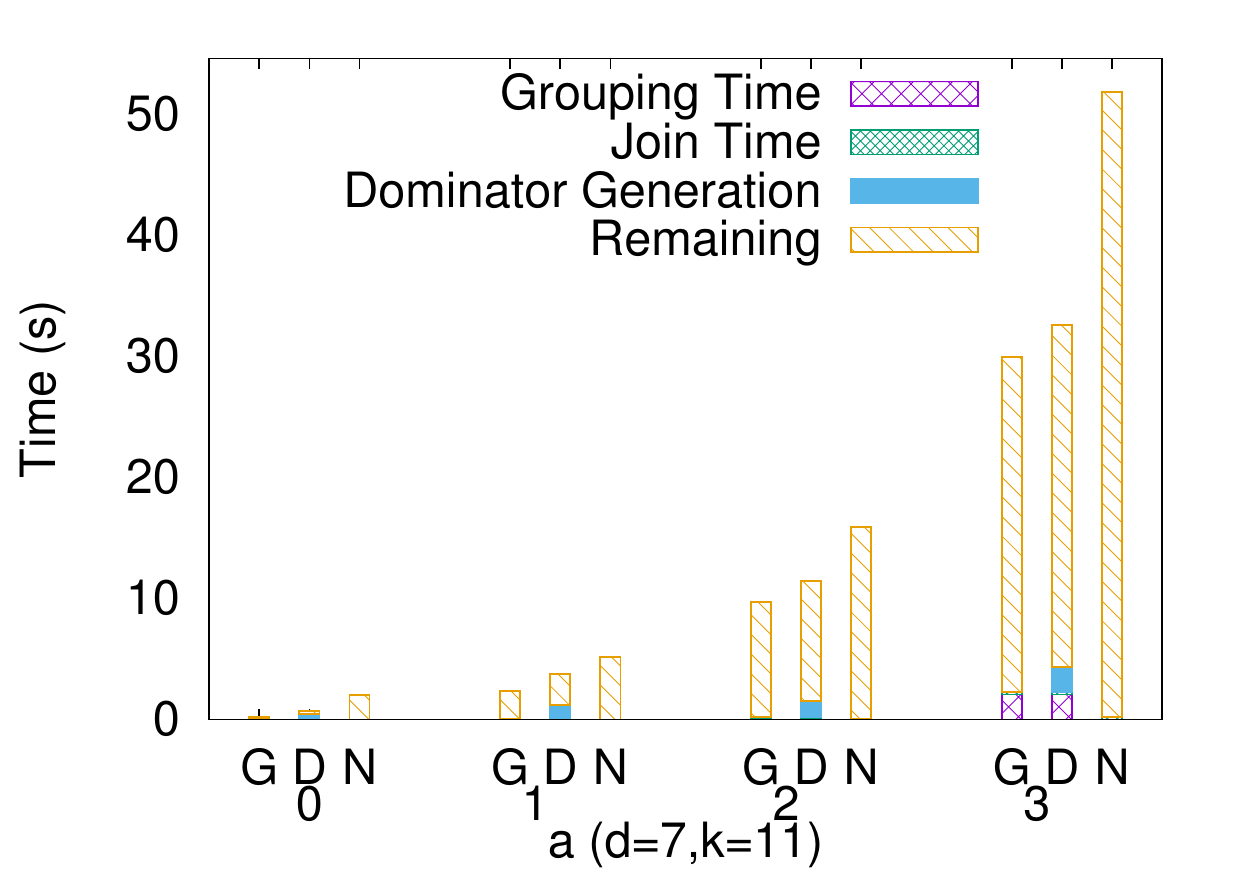}
		\label{subfig:suma}
	}
	\subfloat[Dimensionality.]
	{
		\includegraphics[width=\subfigwidth]{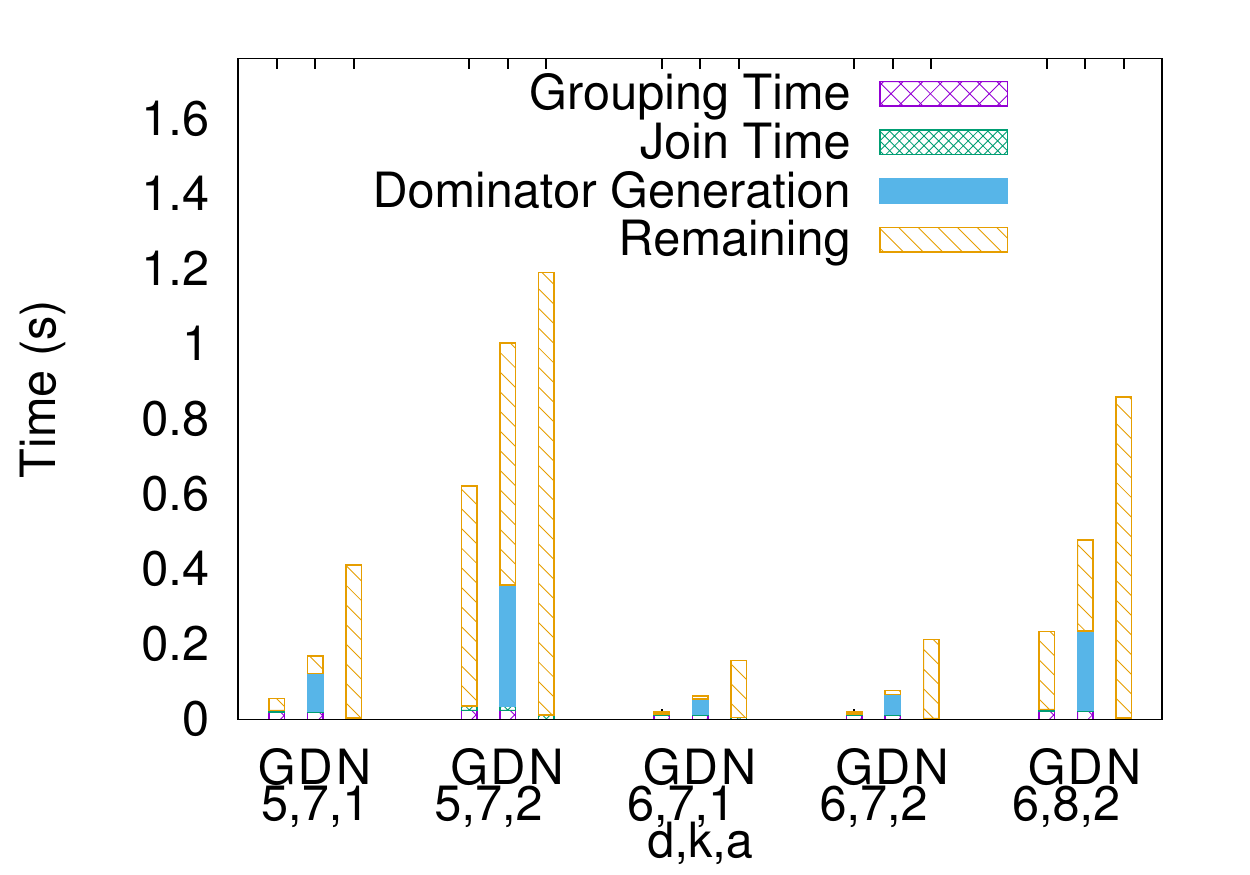}
		\label{subfig:sumd}
	}
	\figcaption{Effect of dimensionality.}
	\label{fig:sumd}
\end{figure}

\begin{figure}[t]
	\centering
	\subfloat[Effect of $g$.]
	{
		\includegraphics[width=\subfigwidth]{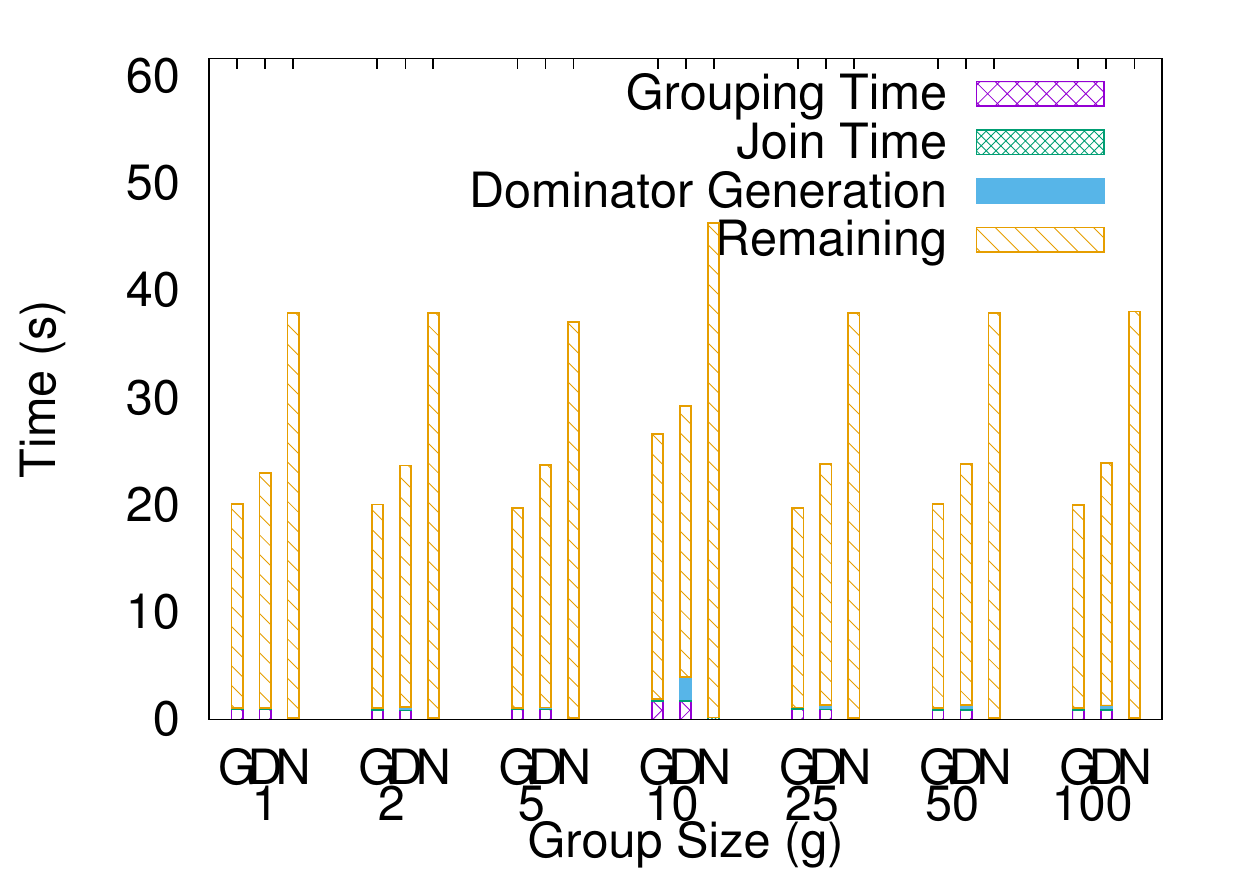}
		\label{subfig:sumg}
	}
	\subfloat[Effect of $n$.]
	{
		\includegraphics[width=\subfigwidth]{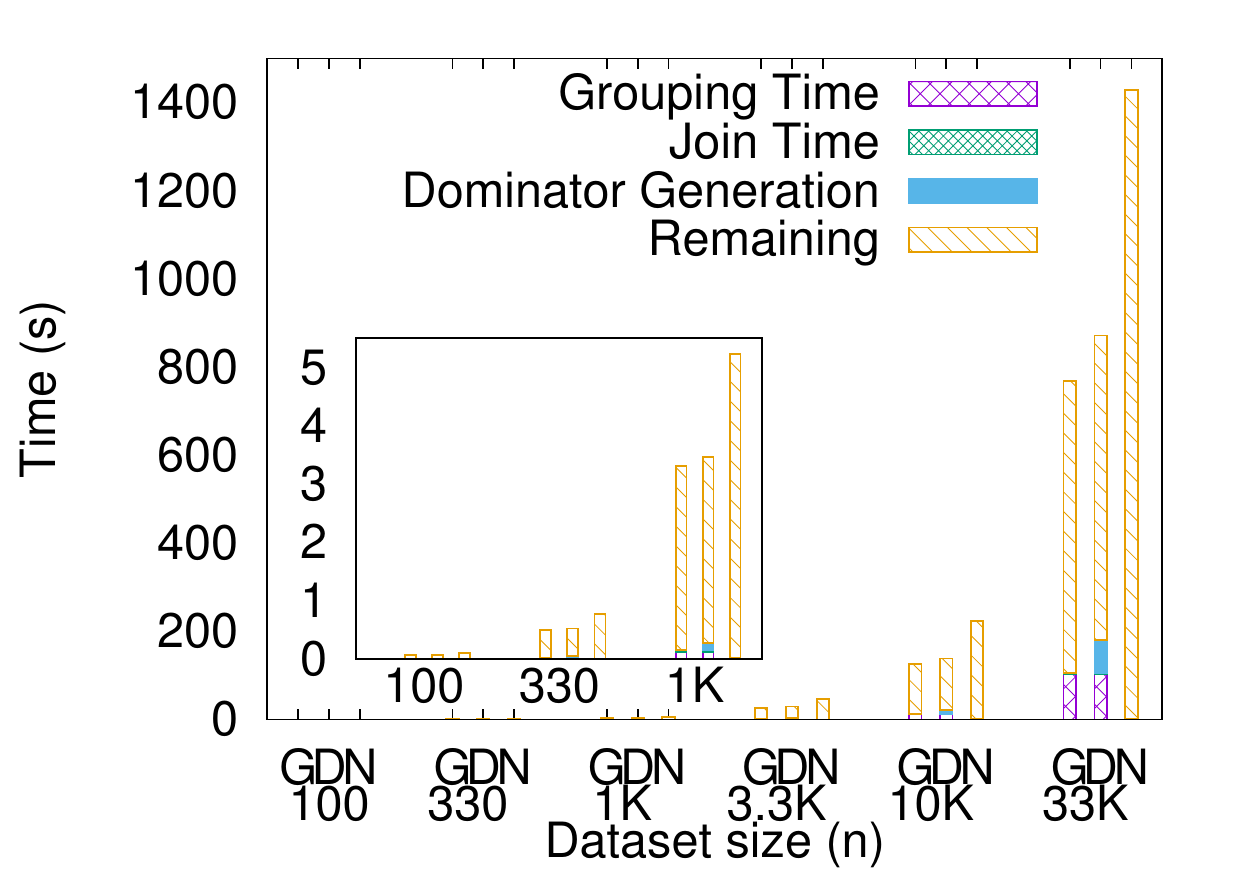}
		\label{subfig:sumn}
	}
	\figcaption{Scalability.}
	\label{fig:sums}
\end{figure}

\subsubsection{Effect of Number of Join Groups}
\label{sec:g}

Fig.~\ref{subfig:sumg} shows the effect of number of join groups.  Note that
when $g=1$, the join reduces to a Cartesian product.  This case can be handled
specially as explained in Sec.~\ref{sec:cartesian}.  When $g$ is low, the
number of \kdom skylines is low since there are more chances of a tuple being
$k$-dominated by another one in the same group.  Thus, there are less number of
$\sn$ tuples (none at $g=1$) and their dominators.  On the other hand, when $g$
is high, the size of the joined relation (which is $n^2/g$) decreases.
Therefore, there are two opposing effects on the running time.  Empirically,
the running times are the highest at medium values.

\subsubsection{Effect of Dataset Size}
\label{sec:n}

The next experiment varies the size of the base relations, $n$.  Note that with
increase in $n$, the size of the joined relation increases quadratically
($O(n^2)$).  Consequently, as visible in Fig.~\ref{subfig:sumn}, the running
time increases drastically.  The scalability of the grouping algorithm, in
particular, as well as the dominator-based algorithm, is sub-linear in the size
of the joined relation, though.

\subsubsection{Effect of Type of Data Distribution}
\label{sec:t}

\begin{figure}[t]
	\centering
	\subfloat[Data type.]
	{
		\includegraphics[width=\subfigwidth]{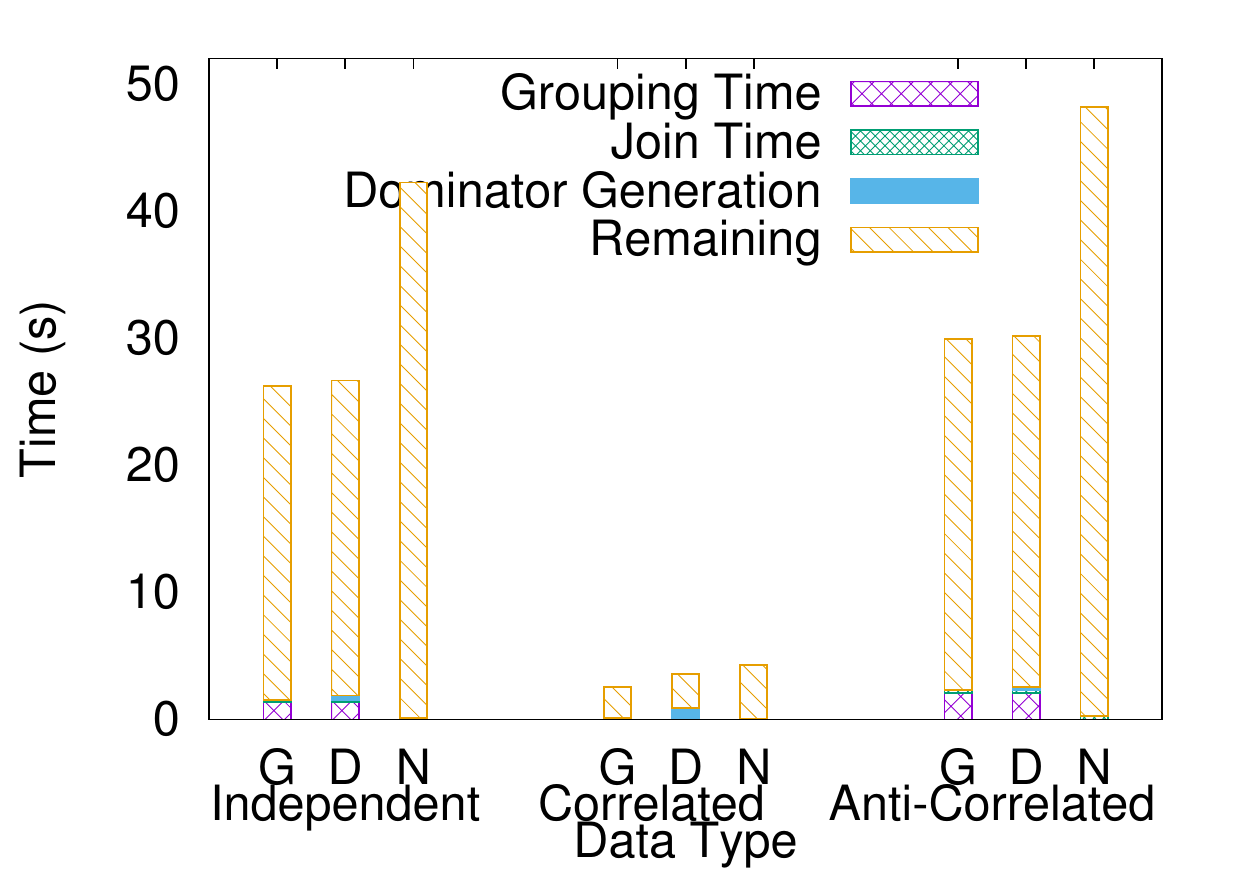}
		\label{subfig:sumt}
	}
	\figcaption{Type of data distribution.}
	\label{fig:sumt}
\end{figure}

The final set of experiments on the aggregated attributes measures the effect
of type of data distribution.  Fig.~\ref{fig:sumt} shows that correlated
datasets are the easiest to process due to higher chances of domination of a
tuple by another tuple, thereby resulting in lesser number of skylines.  The
anti-correlated datasets are the most time consuming due to the opposing
effect.  The independent datasets are mid-way.

\comment{

\subsection{Aggregate with Minimum}

\begin{verbatim}
A
D6_A1
D7_A2
Dimension_2
Group_Size
N
Data_Type
\end{verbatim}

\begin{figure}[t]
	\centering
	\subfloat[Effect of $a$.]
	{
		\includegraphics[width=\subfigwidth]{Selected_Min_Agg/A}
		\label{subfig:mina}
	}
	\subfloat[Effect of $k$.]
	{
		\includegraphics[width=\subfigwidth]{Selected_Min_Agg/D6_A1}
		\label{subfig:mink}
	}
	\figcaption{Min: Effect of dimensionality.}
	\label{fig:mind}
\end{figure}

\begin{figure}[t]
	\centering
	\subfloat[Effect of $a$.]
	{
		\includegraphics[width=\subfigwidth]{Selected_Min_Agg/Dimension_2}
		\label{subfig:mind}
	}
	\subfloat[Effect of $k$.]
	{
		\includegraphics[width=\subfigwidth]{Selected_Min_Agg/Group_Size}
		\label{subfig:ming}
	}
	\figcaption{Min: Effect of dimensionality.}
	\label{fig:mind}
\end{figure}

\begin{figure}[t]
	\centering
	\subfloat[Effect of $a$.]
	{
		\includegraphics[width=\subfigwidth]{Selected_Min_Agg/N}
		\label{subfig:minn}
	}
	\subfloat[Effect of $k$.]
	{
		\includegraphics[width=\subfigwidth]{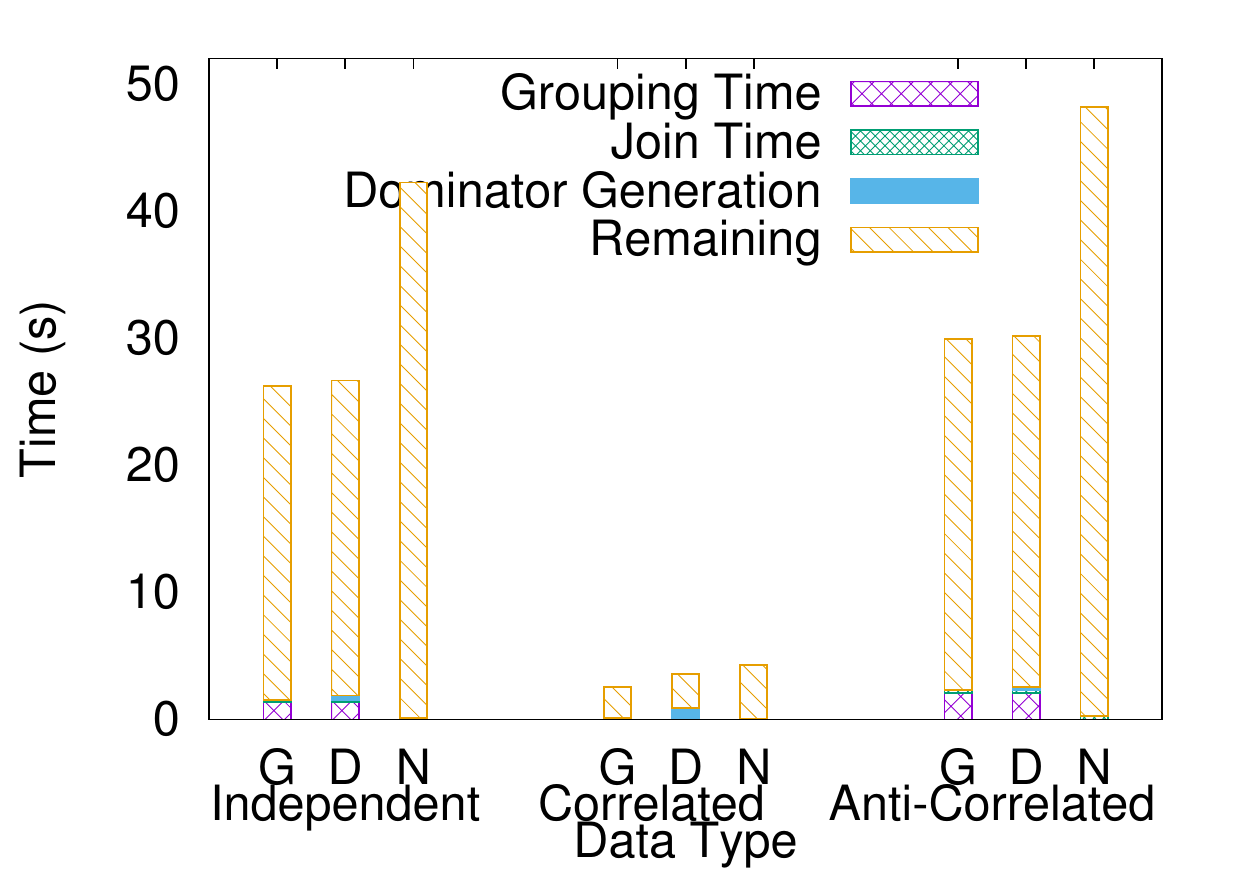}
		\label{subfig:mint}
	}
	\figcaption{Min: Effect of dimensionality.}
	\label{fig:mind}
\end{figure}

\subsubsection{Effect of Dimensionality}

This includes $d$, $k$ and $a$.

\subsubsection{Effect of Dataset Size}

$n$.

\subsubsection{Effect of Data Type}

\subsubsection{Effect of Number of Groups}

$g$.

}

\subsection{No Aggregation}

The next set of experiments target the scenarios where no aggregation over the
skyline attributes is done.

\comment{

\begin{verbatim}
D5
Dimension_2
Group_Size
N
Data_Type
\end{verbatim}

}

\subsubsection{Effect of Dimensionality}

\begin{figure}[t]
	\centering
	\subfloat[Effect of $k$.]
	{
		\includegraphics[width=\subfigwidth]{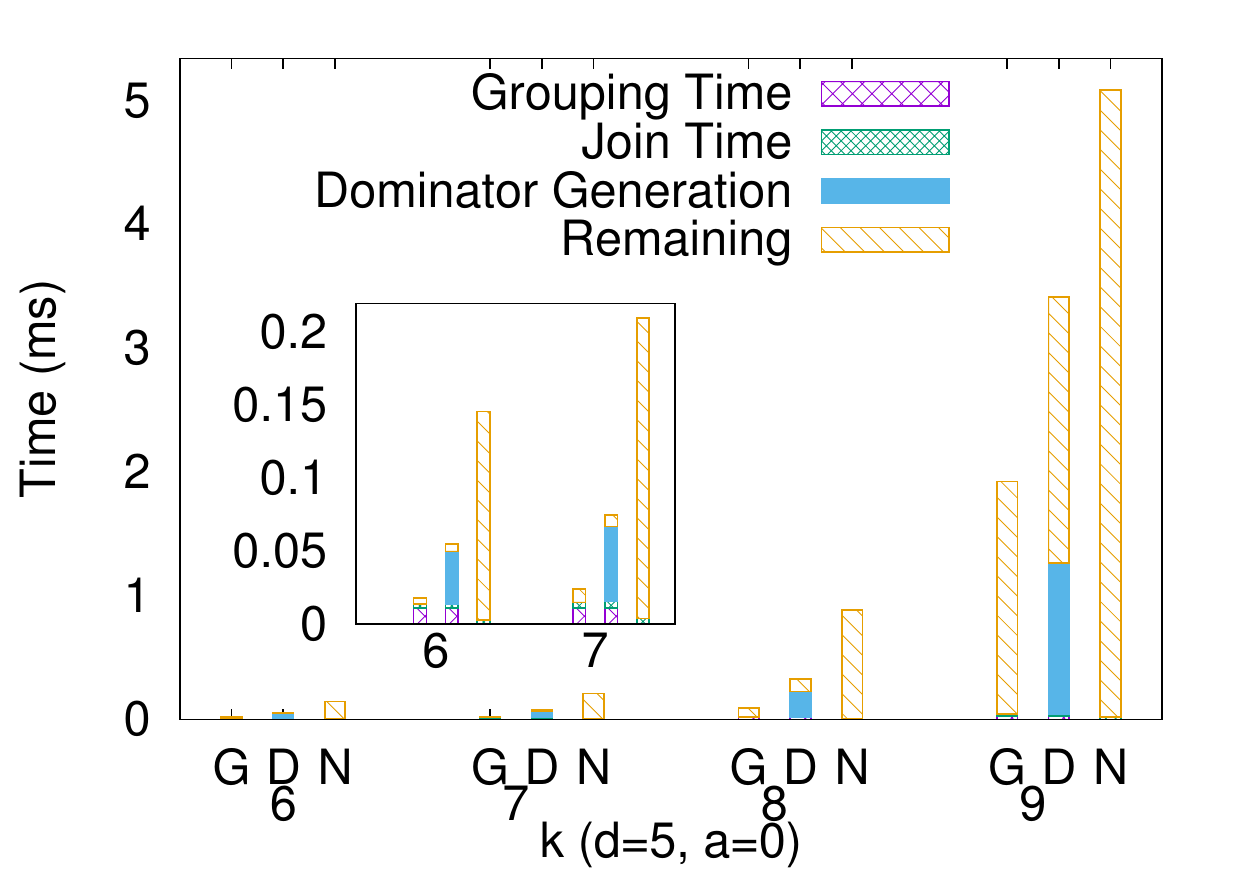}
		\label{subfig:nok}
	}
	\subfloat[Effect of $d$.]
	{
		\includegraphics[width=\subfigwidth]{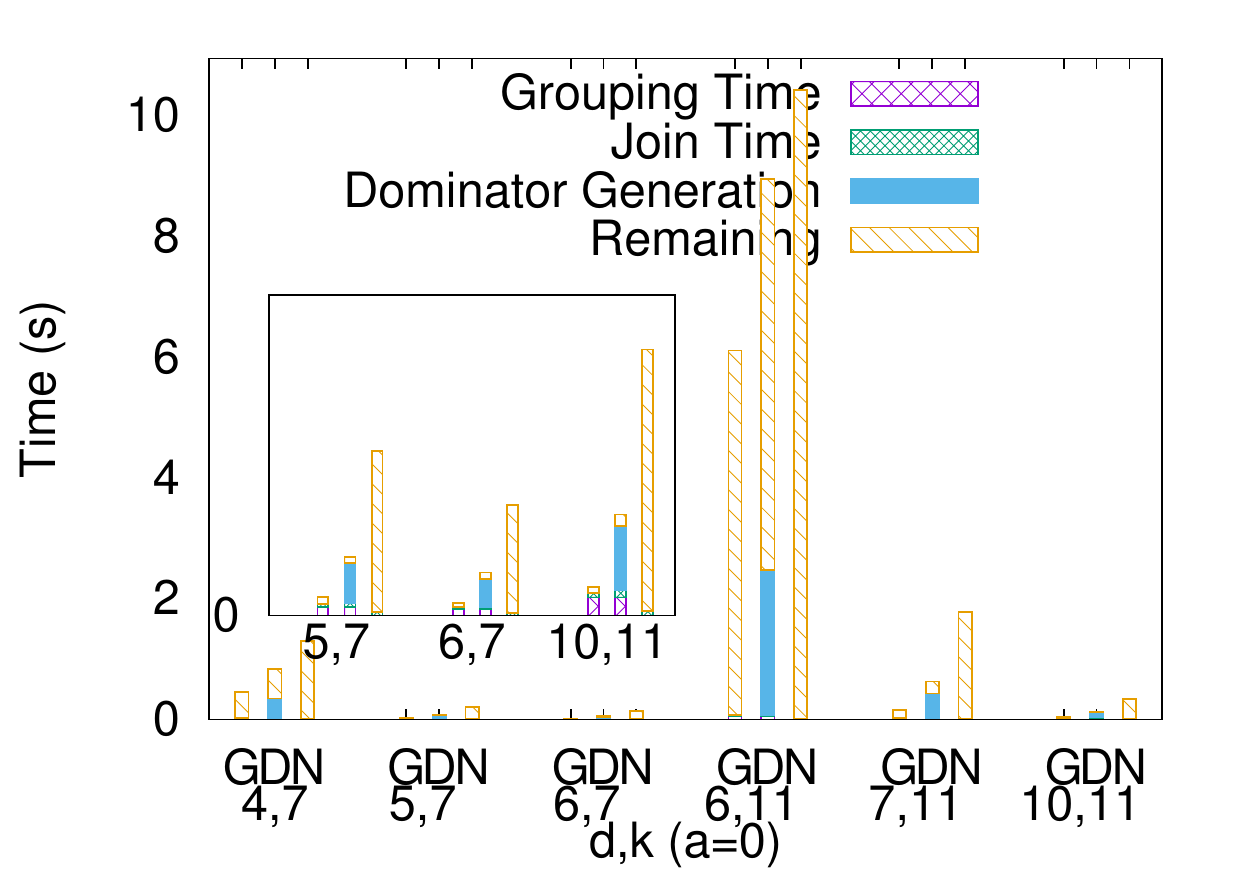}
		\label{subfig:nod}
	}
	\figcaption{Effect of dimensionality (no aggregation).}
	\label{fig:nodimen}
\end{figure}

Fig.~\ref{subfig:nok} shows the effect of $k$ when $d = 5$.  Note that since $a
= 0$, the possible values of $k$ range from $d+1 = 6$ to $2d - 1 = 9$.

Similar to the case with aggregate attributes, the running time increases
sharply with $k$.  The grouping algorithm performs the best while the na\"ive is
the worst.  The dominator-based algorithm suffers since the time spent in
finding the dominators is too large and is not sufficiently compensated later.
Interestingly, since the join time is constant for the na\"ive algorithm
irrespective of the value of $k$, the proportion of time spent in joining is
much higher for lower $k$.

Fig.~\ref{subfig:nod} holds $k$ constant and varies $d$ over two settings.
When $k$ is fixed and $d$ increases, the values of $\koned$ and $\ktwod$
decrease.  This results in faster grouping time.  The dominator sets are also
computed faster.  Thus, the the overall time decreases.

\begin{figure}[t]
	\centering
	\subfloat[Effect of $k$.]
	{
		\includegraphics[width=\subfigwidth]{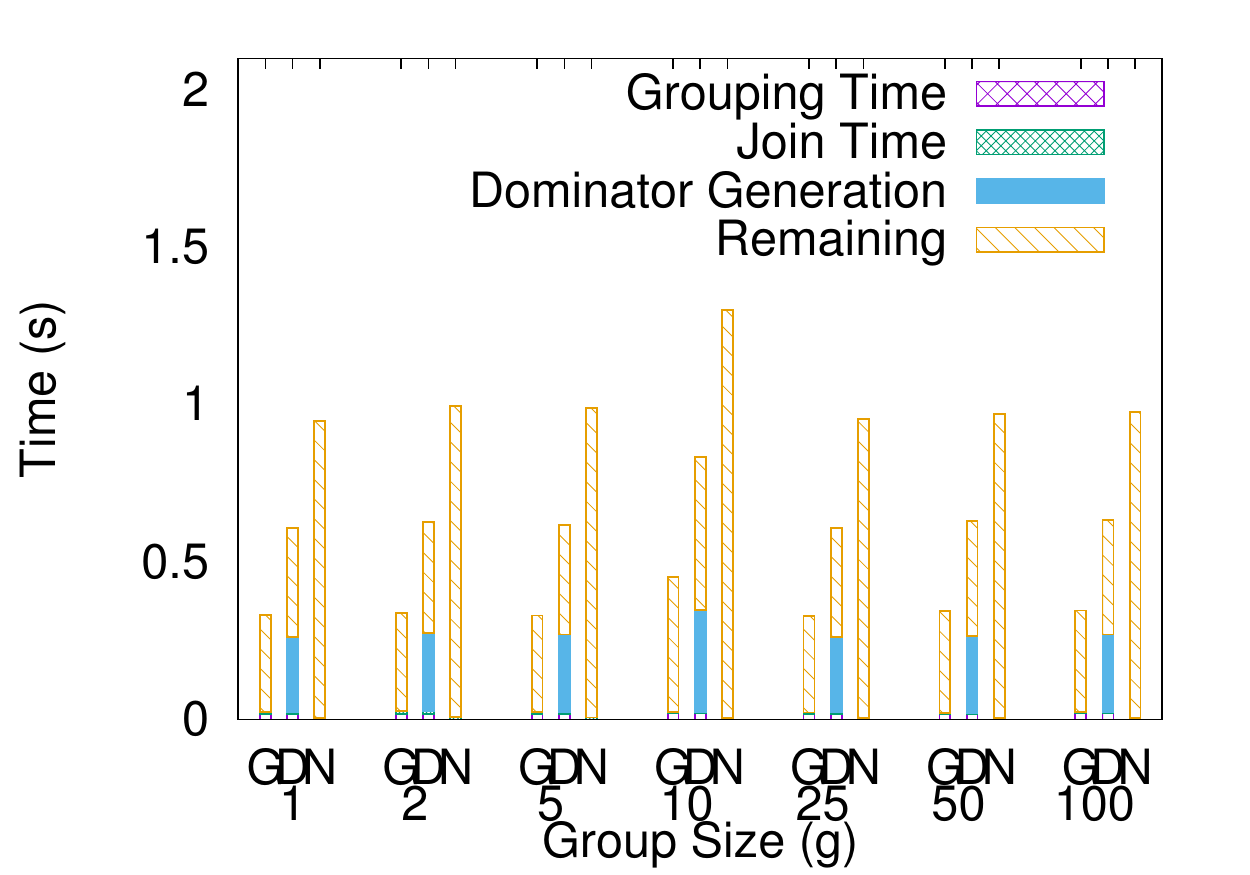}
		\label{subfig:nog}
	}
	\subfloat[Effect of $d$.]
	{
		\includegraphics[width=\subfigwidth]{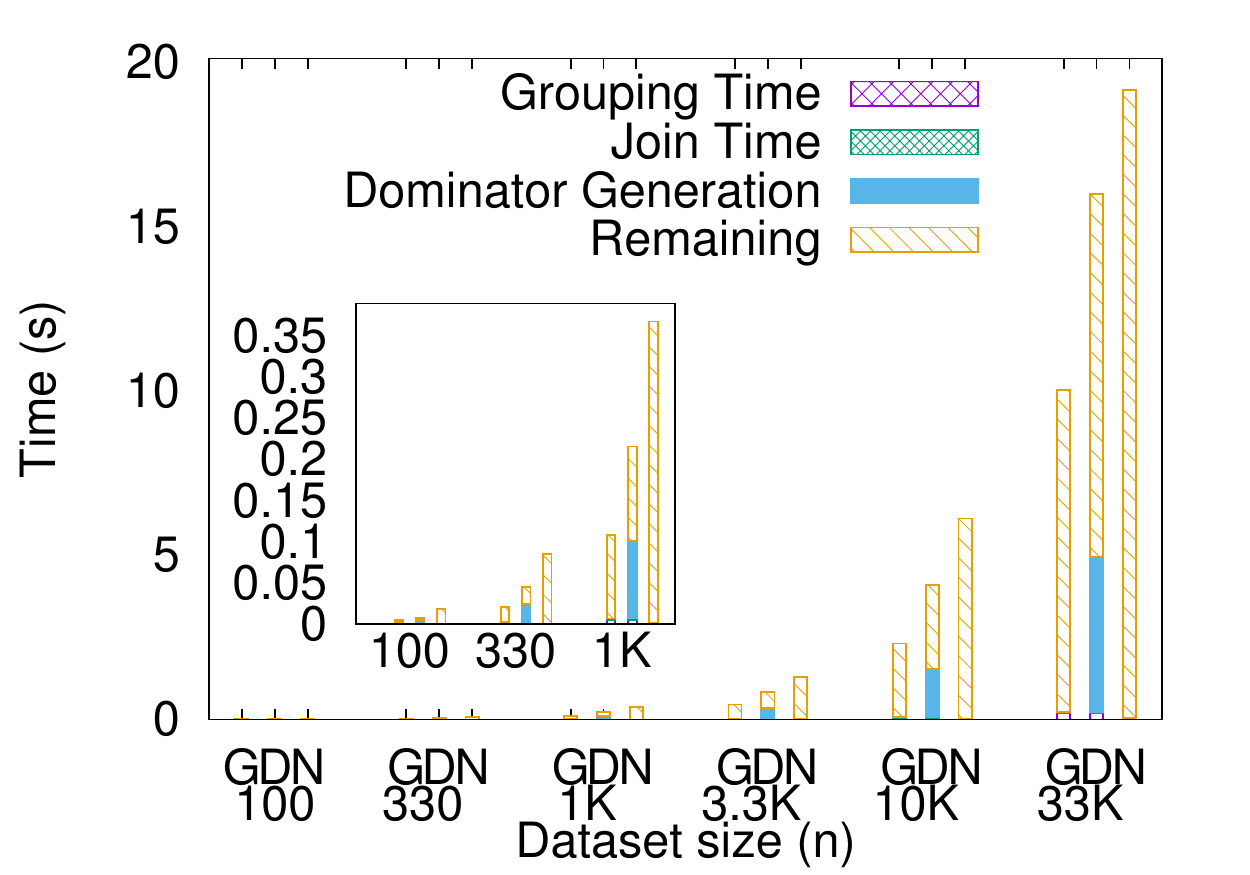}
		\label{subfig:non}
	}
	\figcaption{Scalability (no aggregation).}
	\label{fig:nos}
\end{figure}

\subsubsection{Effect of Number of Join Groups}

As explained earlier in Sec.~\ref{sec:g}, the number of join groups has two
opposing effects on the running time.  Therefore, as shown in
Fig.~\ref{subfig:nog}, the results are similar.  Note that the values of $d$
and $k$ here are $d = 4$ and $k = 7$.

\subsubsection{Effect of Dataset Size}

Fig.~\ref{subfig:non} shows that the running time increases drastically with
$n$, although the scalability is sub-linear in the size of the joined relation,
i.e., $n^2$.  The largest dataset size we tested was for $n = 33,000$, leading
to a massive size of over $10^8$ for the joined relation.  The fact that the
grouping algorithm produces the result in less than 20\,s for this case
establishes the practicality of the algorithms.

\subsubsection{Effect of Type of Data Distribution}

The effect of type of data distribution (Fig.~\ref{fig:not}) is similar to that
in Sec.~\ref{sec:t} with the anti-correlated requiring the largest amount of
time and correlated the least.

\begin{figure}[t]
	\centering
	\subfloat[Data type.]
	{
		\includegraphics[width=\subfigwidth]{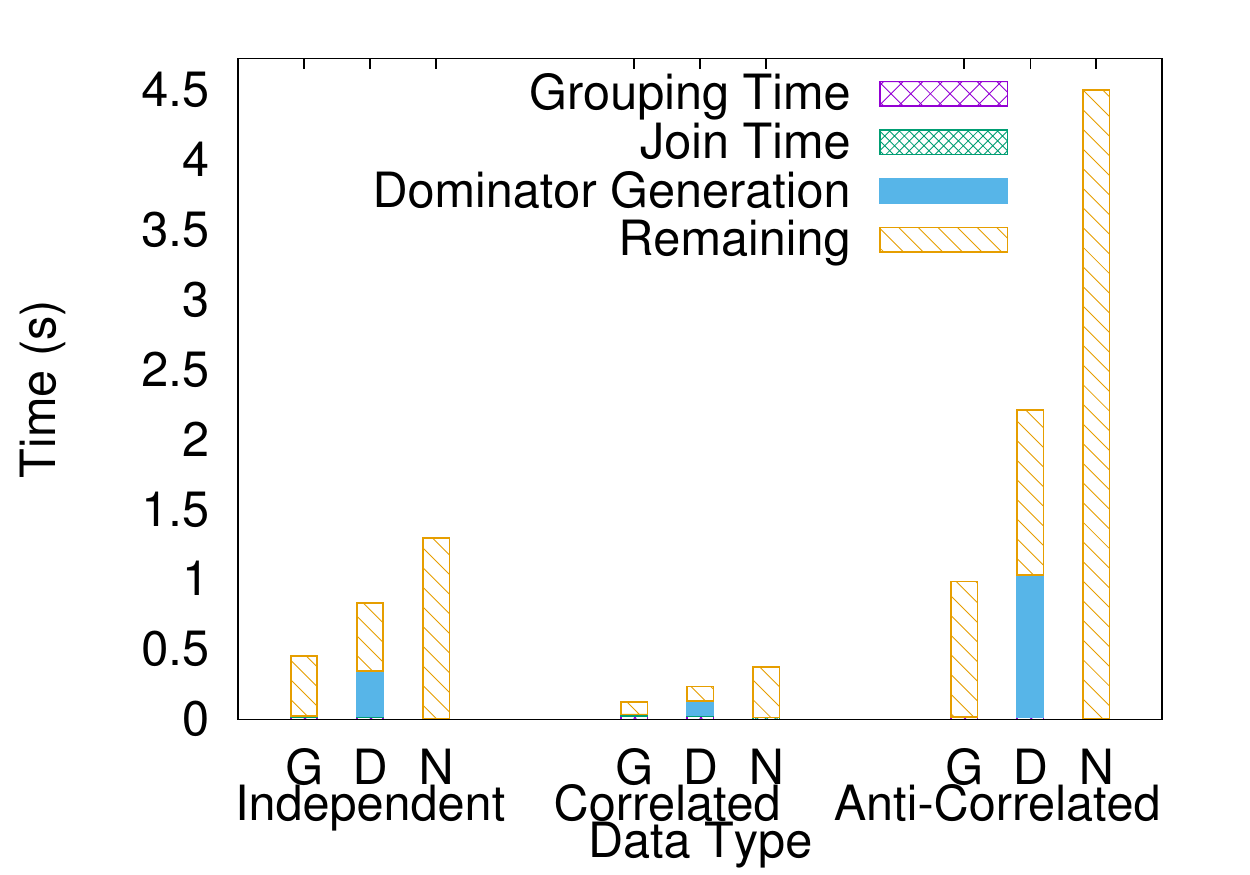}
		\label{subfig:not}
	}
	\figcaption{Type of data distribution (no aggregation).}
	\label{fig:not}
\end{figure}

\subsection{Finding $k$}

The third and final set of experiments deals with finding the value of $k$
given a threshold $\delta$ of requisite number of \kdom skylines.  Aggregate
attributes are not used.

\comment{

\begin{verbatim}
Data_Type
Delta_3300
Dimension_10000
Group_Size
N
\end{verbatim}

}

\subsubsection{Effect of Threshold $\delta$}

Fig.~\ref{subfig:kdelta} shows the effect of varying $\delta$ values.  The
dimensionality is held fixed at $d = 5$ with $a = 0$.  The possible values of
$k$, therefore, range from $d+1 = 6$ to $2d = 10$.  The size of each base
relation is $n = 3,300$ with $g = 10$, thereby resulting in more than $10^6$
joined tuples (as listed in Table~\ref{tab:param}).

With increasing $\delta$, the na\"ive algorithm requires larger running times
since it keeps iterating over $k$.  The range-based search also iterates over
the values of $k$, although it avoids the costly full \kdom skyline computation
in the intermediate steps, if possible.  When $\delta$ is very large, it falls
outside the upper bounds for most values of $k$ and, hence, the algorithm runs
very fast.  In this experiment, when $\delta \geq 10,000$, the largest possible
$k = 10$ is returned as the answer.

For low values of $\delta$, the iterations of both na\"ive and range-based
algorithms stop early.  The answer for $\delta = 10$ is $k = 8$ while that for
$\delta = 100$ and $\delta = 1,000$ are both $k = 9$.

The binary search method is faster even for medium values of $\delta$ since it
avoids the iterative procedure and quickly finds the desired value of $k$.
Overall, it is always the fastest algorithm.

\begin{figure}[t]
	\centering
	\subfloat[Effect of $\delta$.]
	{
		\includegraphics[width=\subfigwidth]{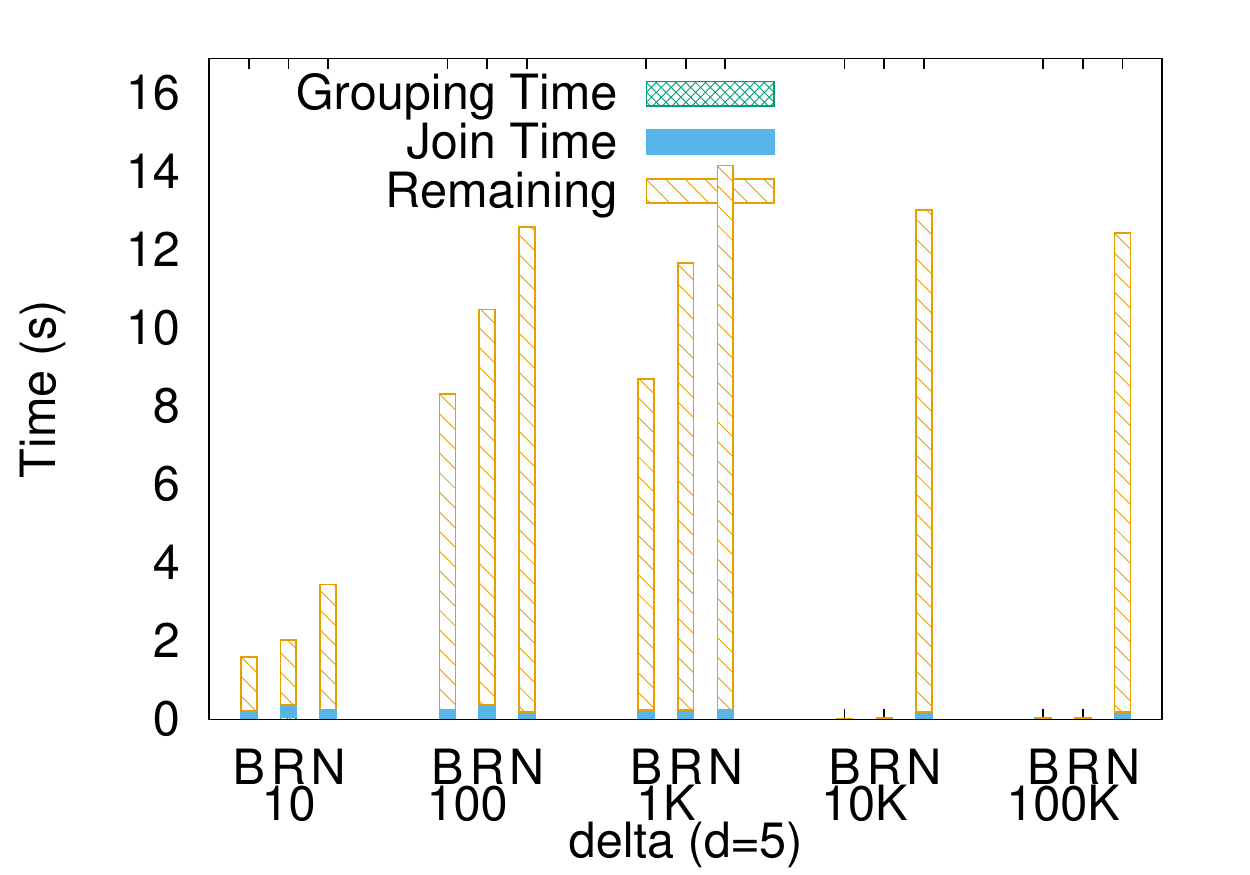}
		\label{subfig:kdelta}
	}
	\subfloat[Effect of $d$.]
	{
		\includegraphics[width=\subfigwidth]{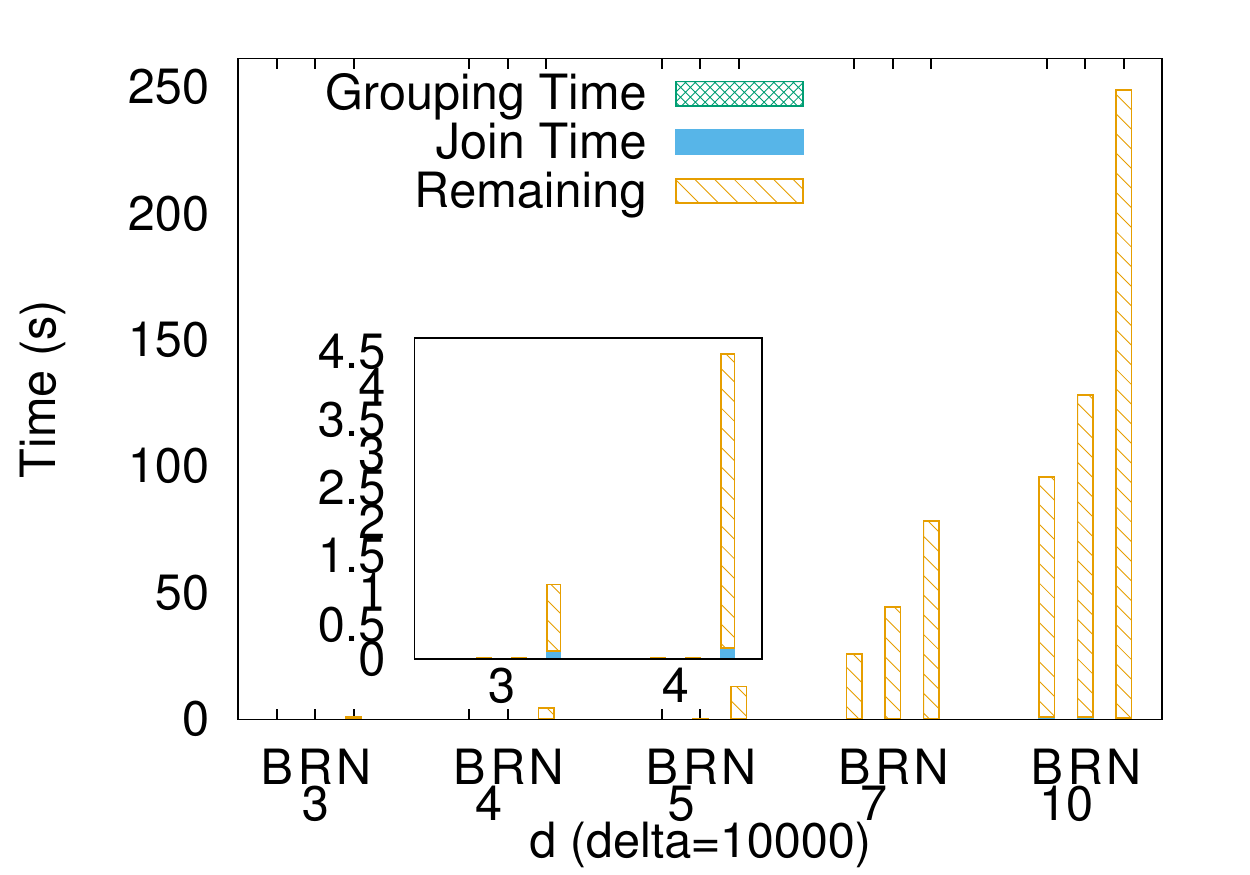}
		\label{subfig:kd}
	}
	\figcaption{Effect of dimensionality (finding $k$).}
	\label{fig:kdeltad}
\end{figure}

\subsubsection{Effect of Dimensionality $d$}

Fig.~\ref{subfig:kd} shows the complementary effect, that of varying
dimensionality $d$ while keeping $\delta$ constant at $10,000$.  When $d$ is
low, the chosen $\delta$ value is too large and the algorithms terminate fast.
When $d$ is increased, the algorithms need to search through a larger range
and, therefore, takes a much longer time.  The binary search method is
consistently the fastest algorithm.  It outperforms the range-based algorithm
by about 1.2-1.5 times while the na\"ive algorithm is slower by a factor of
2-2.5.

\begin{figure}[t]
	\centering
	\subfloat[Effect of $g$.]
	{
		\includegraphics[width=\subfigwidth]{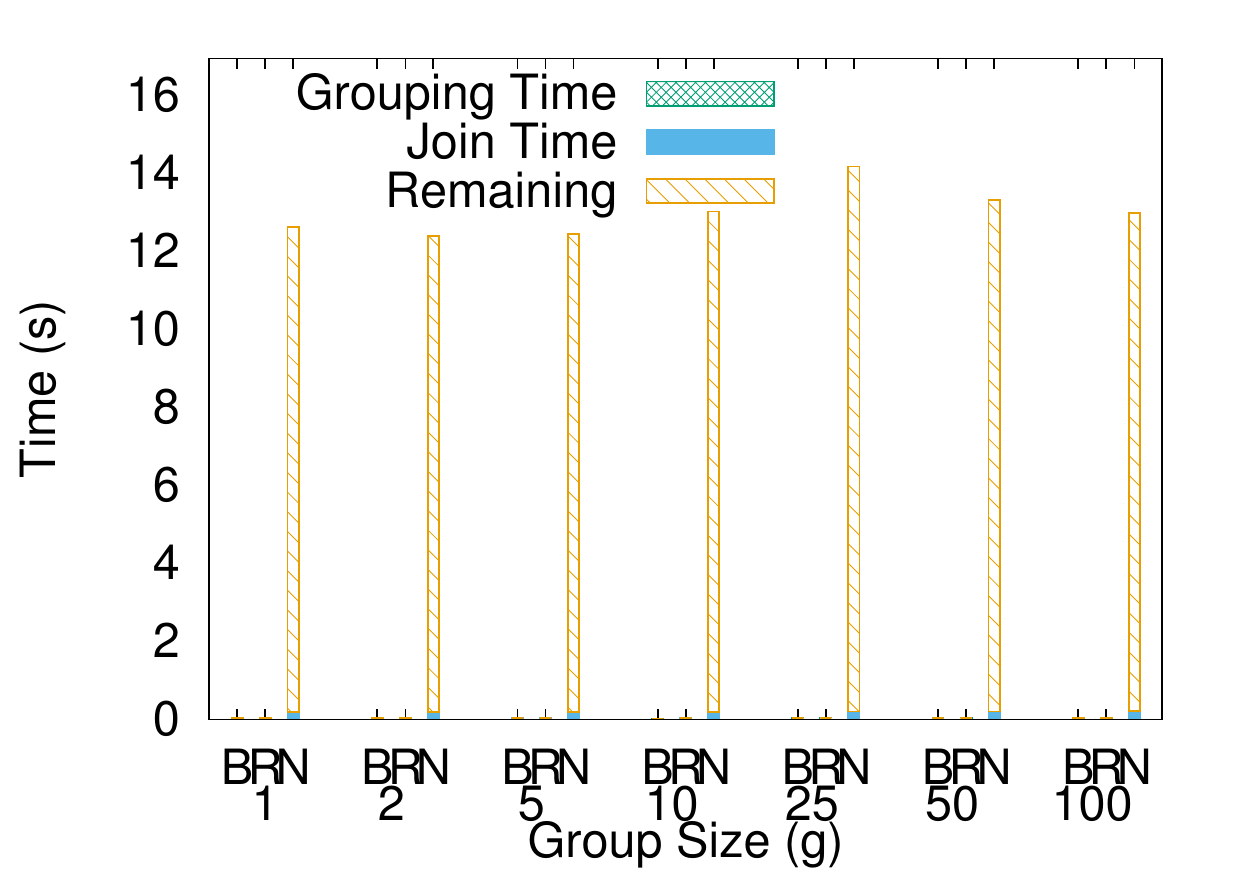}
		\label{subfig:kg}
	}
	\subfloat[Effect of $n$.]
	{
		\includegraphics[width=\subfigwidth]{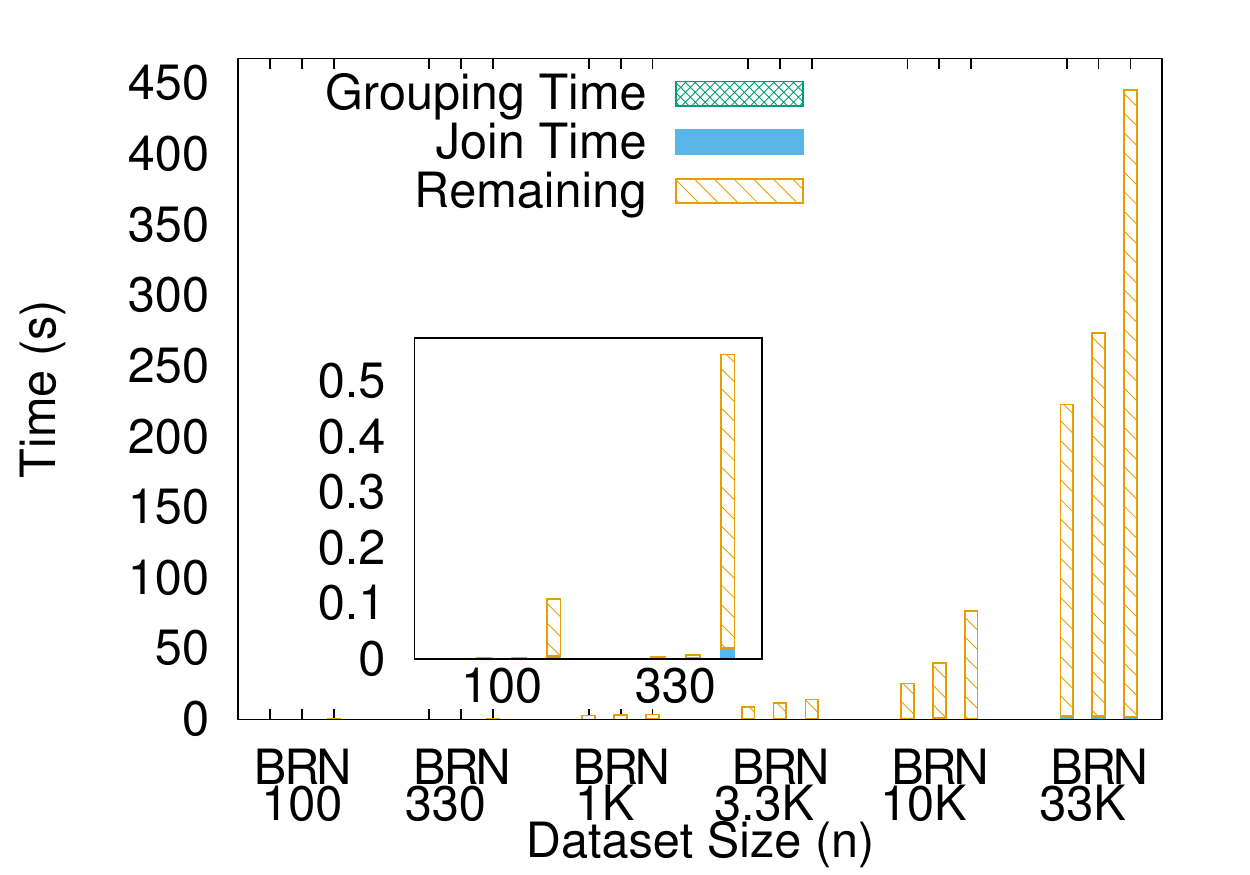}
		\label{subfig:kn}
	}
	\figcaption{Scalability (finding $k$).}
	\label{fig:ks}
\end{figure}

\subsubsection{Effect of Number of Join Groups}

There is no appreciable effect of the number of join groups on the algorithms
for finding $k$ (Fig.~\ref{subfig:kg}).

\subsubsection{Effect of Dataset Size}

Fig.~\ref{subfig:kn} shows the effect of increasing dataset size, $n$.  The
dimensionality and threshold values are kept fixed at $d = 5$ and $\delta =
1,000$ with $g = 10$.

For very low values of $n$ (up to $1,000$), the threshold is too high, and the
maximum possible $k = 10$ is required to satisfy the threshold $\delta$.

With increasing $n$, the running time increases due to increasing \kdom skyline
computation times.  Even for larger $n$, the values of $k$ are towards the
higher end.  Therefore, the binary search algorithm is the most suitable one
for finding $k$.

\begin{figure}[t]
	\centering
	\subfloat[Data type.]
	{
		\includegraphics[width=\subfigwidth]{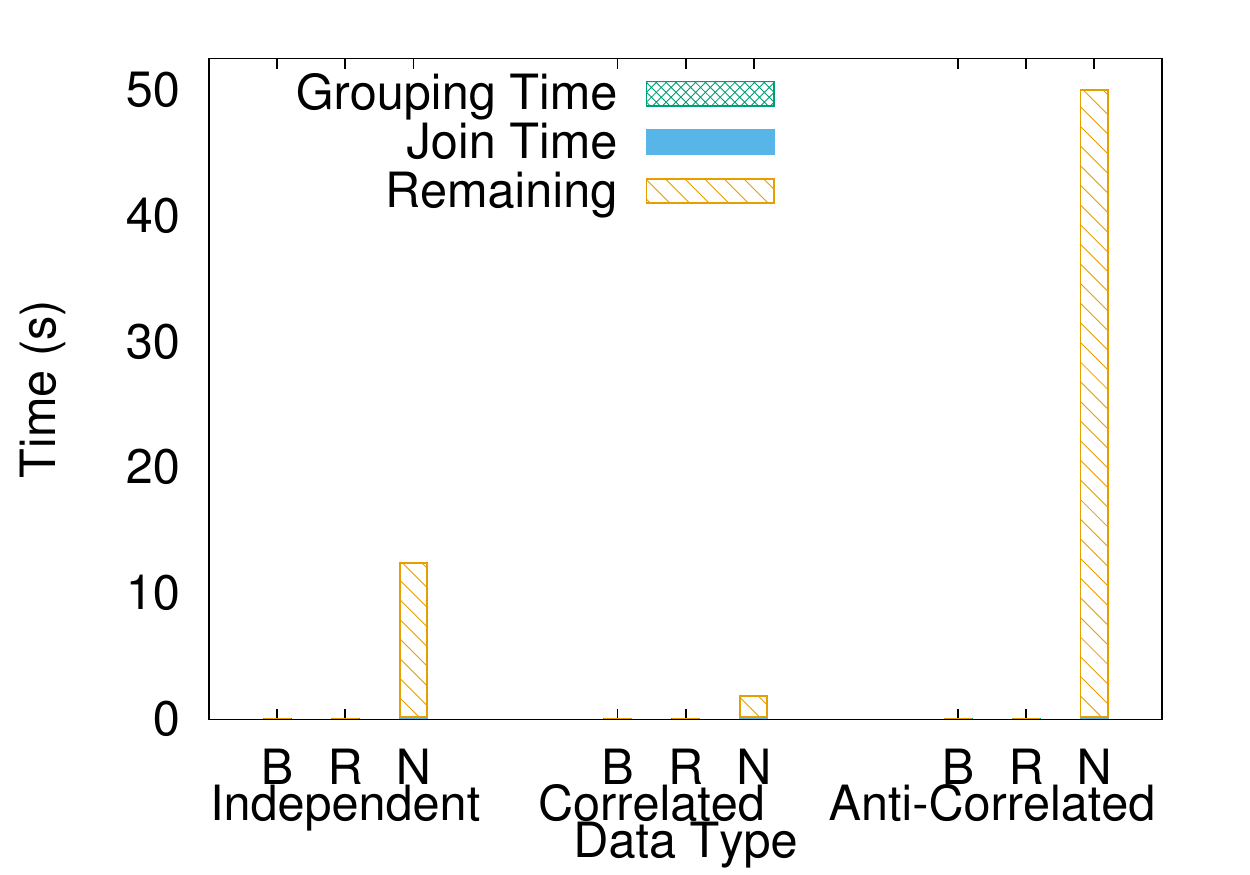}
		\label{subfig:kt}
	}
	\figcaption{Type of data distribution (finding $k$).}
	\label{fig:kt}
\end{figure}

\subsubsection{Effect of Data Type}

The effect of type of data distribution is as expected with correlated being
the fastest and anti-correlated the slowest (Fig.~\ref{subfig:kt}).

\subsection{Real Dataset}
\label{sec:real}

To test our algorithms on real data, we collected information about various
attributes on domestic flights in India from \url{www.makemytrip.com}.  The
first base table contained information about $192$ flights from New Delhi to
$13$ important cities of India while the second base table contained
information about $155$ flights from those $13$ cities to Mumbai.  For each
base table, $5$ attributes were considered: cost, flying time, date change fee,
popularity, and amenities.  The first $2$ attributes were aggregated while the
other $3$ were used as local attributes.  The join was based on the equality of
the intermediate city.  Thus, each tuple in the joined relation contained $3 +
3 + 2 = 8$ attributes.  The size of the joined relation was $2649$.


We ran experiments on $k = 6, 7, 8$.  The results are summarized in
Fig.~\ref{fig:real}.
The grouping-based algorithm performed the best followed by the dominator-based
method and the na\"ive algorithm.  All the results were produced in
milliseconds highlighting the practicality of the methods.

\subsection{Summary of Experiments}

The experiments show that the grouping algorithm consistently outperforms the
other methods in solving the KSJQ queries.  For finding the right value of $k$,
the binary search algorithm turns out to be the best method always.

\section{Conclusions}
\label{sec:conc}

In this paper, we proposed a novel query, $k$-dominant skyline join query
(KSJQ), that incorporates finding \kdom skylines over joined relations where
the attributes may be aggregated as well.  We analyzed certain optimizations
for the query and used them to design efficient algorithms.  In addition, given
the number of final skylines sought, we also proposed efficient algorithms to
find the right value of $k$.

In future, we would like to extend the algorithms to work in parallel,
distributed and probabilistic settings.

\begin{figure}[!tb]
	\centering
		\includegraphics[width=\subfigwidth]{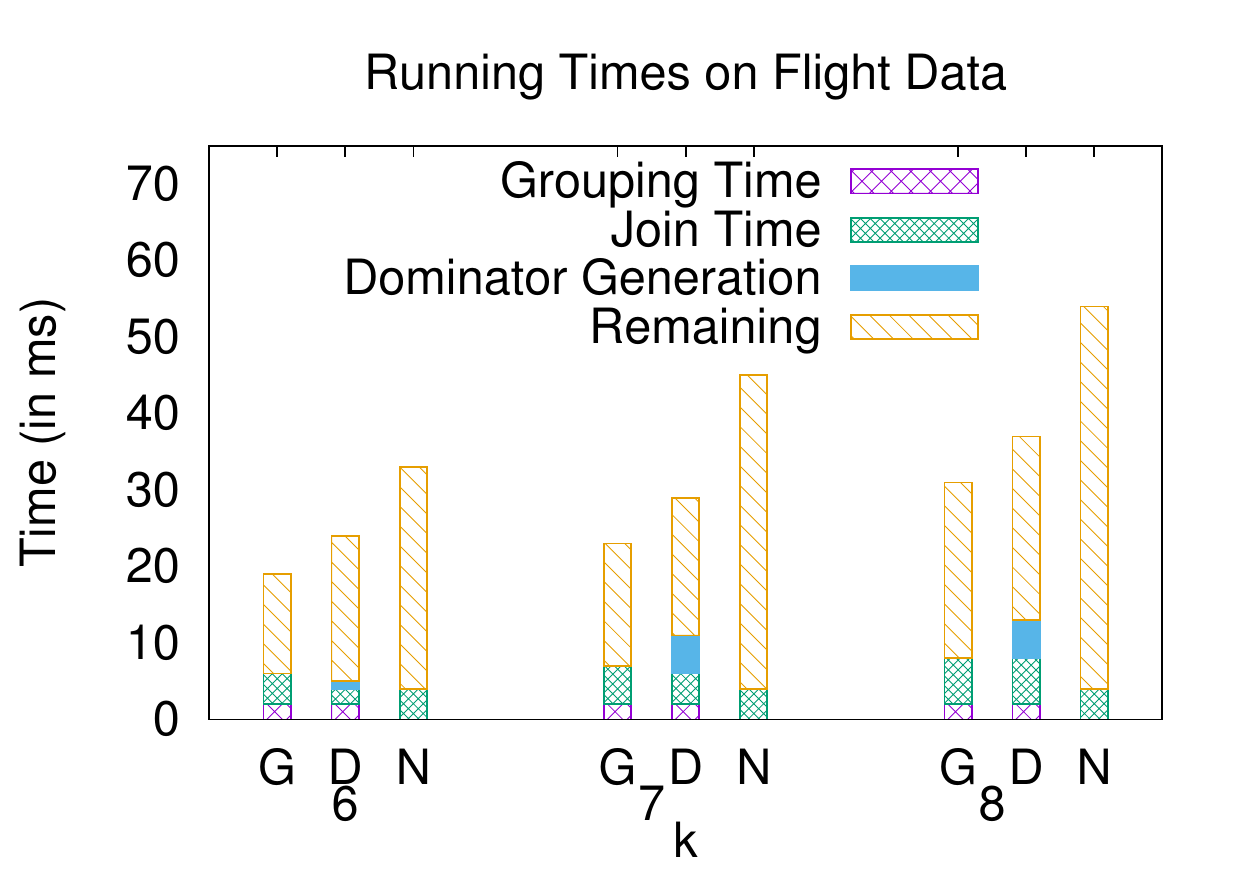}
	\figcaption{Real data.}
	\label{fig:real}
\end{figure}


{\small
\bibliographystyle{abbrv}
\balance
\bibliography{papers}
}

\end{document}